\documentclass[a4paper,11pt,oneside]{article}
\usepackage{amsmath,amsthm,amsfonts,latexsym,amscd,amssymb, graphicx, color,enumerate, longtable, booktabs,array}

\newcolumntype{N}{>{\centering\arraybackslash}p{1cm}}
\newcolumntype{L}{>{\arraybackslash}p{6cm}}          
\newcolumntype{R}{>{\arraybackslash}p{6cm}}          
\DeclareUnicodeCharacter{200B}{}
\usepackage[hidelinks]{hyperref}
\newtheorem{thm}{Theorem}[section]
\newtheorem{dfn}[thm]{Definition}
\newtheorem{rmk}[thm]{Remark}
\newtheorem{exl}[thm]{Example}

\newtheorem{propo}[thm]{Proposition}
\newtheorem{coro}[thm]{Corollary}

\newtheorem{lem}[thm]{Lemma}

\begin{document}
	\title{\bf { On concatenation of matrices for reversible linear codes over a finite commutative ring and applications to DNA codes} }  
	\author{\bf Avanish Kumar Chaturvedi$^{a}$\footnote{Corresponding author} and \bf Satyadeep Pandey$^{b}$} 
	\date{Department of Mathematics, University of Allahabad \\ Prayagraj-211002, India \\ {\footnotesize  $^{a}$akchaturvedi.math@gmail.com, achaturvedi@allduniv.ac.in\\
			$^{b}$ pandeysatyam14082000@gmail.com}}     	
	\maketitle   
	\begin{abstract}
		In this paper, we develop a generalized framework for constructing reversible linear, reversible self dual and reversible DNA codes using a matrix-theoretic approach based on involutory matrices. The proposed concatenation scheme gives a large class of generator matrices and yields codes with good parameters. The construction is carried out at the level of linear codes and then extended to DNA codes. Using a matrix product approach, we provide a unified method for analysis and proof. Further, we resolve an open problem raised by Oztas et.al. and also we correct and improve some results of them.
	\end{abstract}
	\section{Introduction}Deoxyribonucleic acid (DNA) plays a fundamental role in storing and transmitting genetic information in living organisms. A DNA molecule consists of two antiparallel strands forming a double helix, where each strand is composed of four nucleotides: Adenine (A), Thymine (T), Cytosine (C), and Guanine (G). These nucleotides pair according to Watson--Crick complementarity, with A pairing with T and C pairing with G. Each strand possesses an inherent directionality determined by its $5'$ and $3'$ ends, and biological processes rely heavily on both strand orientation and nucleotide complementarity. Owing to these structural properties, DNA sequences are naturally associated with notions of reversal and complement, which have become central constraints in the design of DNA codes. 
	
	The computational potential of DNA was first demonstrated in \cite{Adleman}, in the context of biomolecular computing, where DNA strands were used to solve a combinatorial problem, the traveling salesman problem. Other applications like breaking DES in cryptography \cite{AdlemanDES}, \cite{Boneh}, and storage applications \cite{Mansuripur} have also been explored. However, practical deployment of DNA-based systems requires careful sequence design to avoid undesirable hybridization and secondary structures \cite{Mansuripur}. Consequently, the construction of DNA codes satisfying constraints such as Hamming distance constraint, reversibility, complementarity, and balanced nucleotide distribution, that is, GC-content constraint, has emerged as an important research area. Multiple construction approaches for DNA codes prioritizing one or the other constraints, have been explored in \cite{Abolution}, \cite{Frutos}, \cite{Gaborit}, \cite{King}, \cite{Smith} and \cite{Song}. 
	
	In \cite{Oztas}, authors utilized the concept of double-reflected matrices, wherein a generator matrix is concatenated with its row-wise and column-wise reflection to construct reversible linear codes without imposing restrictive algebraic structures such as cyclicity or polynomial representations. Moreover, it allows reversible DNA codes to be derived from arbitrary full-rank matrices over suitable finite fields, thereby significantly broadening the design space of reversible DNA codes.
	
	By the motivation of matrix theoretic approach as mentioned in \cite{Oztas}, in the present work, we introduce a new matrix concatenation method for reversible linear codes. In \cite{Oztas} authors used double reflected matrix $M''$ of a generator matrix $M$ in construction of a reversible linear code (see Definition \ref{double reflection definition}). We introduce the idea of $D-$ reflected matrix $M_D'$ with respect to any involutory matrix $D$ for generator matrix $M$ in construction of reversible linear codes (see Definition \ref{D-reflection definition}). We find that $M''$ is a particular case of $M_D'$ if we take involutory matrix $D=J_k$ (see Lemma \ref{product of matrices} (5)). By our method the most of the results as given in \cite{Oztas} are consequence and their proofs are more clear and easy in computation.
	
	In Remark \ref{large number of mds} (2), by one example of a generating matrix over $F_{16}$, using SageMath we compute that there are 256 involutory matrices over $F_{16}$ of order $2\times 2.$ We show that there can be 256 reversible linear codes with respect to $256$ different involutory matrices over $F_{16}$ of order $2\times 2.$ However, by the method used in \cite{Oztas}, we can construct only one reversible linear code. Hence, we show that a strictly larger class of matrices can be concatenated with any full rank matrix $M$ over a finite commutative ring generating reversible linear codes.
	
	This not only enlarges the class of generator matrices but also leads to better code parameters in many cases. By SageMath computations \cite{Sage} we obtain a plenty number of MDS reversible linear codes, that is, $120$ and also $136$ AMDS codes out of $256$. However by the method used in \cite{Oztas} we can construct only one AMDS code.
	
	The organization of this paper is as follows. In Section 2, we present the necessary preliminaries and definitions required throughout the paper. Section 3 introduces the construction of reversible linear codes by a suitable concatenation of generator matrices, using any arbitrary involutory matrix along with an analysis of their structural and distance properties. This leads to strict generalization of the concatenation technique given in \cite{Oztas} and also gives better code parameters (see Example \ref{better example}).
	
	In Section 4, we discuss the application of the concatentation approach to generate reversible codes which are also self dual. If a matrix generates a self dual code then the concatented matrix generates a reversible self dual code. In Section 5, we specialize the construction over finite fields of the form $\mathbb{F}_{4^{2t}}$ and establish the connection between reversibility of linear codes and DNA codes.
	
	Further, we show that it leads to a large number of reversible DNA codes (see Remark \ref{better 2}). Also, we give a counterexample (Example \ref{wrong characterization example}) to show that a result in \cite[Corollary 3]{Oztas} is an incorrect characterization of reversible complement codes. We find that the open problem raised in \cite{Oztas} has false assertion, in support, we provide Example  \ref{open problem example}). Finally, Section 6 concludes the paper with remarks on potential extensions and applications.
	
	\section{Preliminaries}In this section, we recall basic definitions and notation that will be used throughout the paper. 	
	Let $R$ be a finite commutative ring with identity. Following \cite{Oztas}, an \emph{$R$-linear code} $C$ of length $n$ is an $R$-submodule $C \subseteq R^n$.
	
	\begin{dfn}
		Let $A$ and $B$ be two $k\times n$ matrices over a ring $R$. By \emph{concatenation} of $A$ and $B$ we mean a row wise concatenation $[A|B]$, of order $k\times 2n.$
	\end{dfn}
	Recall \cite{Oztas}, for a vector $c = (c_o,c_1,\ldots,c_{n-1}) \in R^n$, the \emph{reverse} of $c$ is defined as $ c^r = (c_{n-1},c_{n-2},\ldots,c_0).$
	\begin{dfn}\cite[Definition 1]{Oztas}
		A linear code $C$ of length $n$ over $R$ is said to be a \emph{reversible} code if for any $c \in C,\:c^r \in C.$
	\end{dfn}
	\begin{dfn}\cite[Definition 2]{Oztas}\label{double reflection definition} Let $M = [a_{ij}]_{m \times n}$ be an $m \times n$ matrix over the ring $R$. Then, the \emph{double-reflected matrix} of $M$ is denoted by $M''$ and defined as $
		M'' = [b_{ij}]_{m \times n},$
		where $b_{ij} = a_{(m-i+1)(n-j+1)},$
		that is,
		\[
		M'' =
		\begin{bmatrix}
			a_{mn} & \cdots & a_{m3} & a_{m2} & a_{m1} \\
			\vdots &        & \vdots & \vdots & \vdots \\
			\vdots &        & \vdots & \vdots & \vdots \\
			a_{2n} & \cdots & a_{23} & a_{22} & a_{21} \\
			a_{1n} & \cdots & a_{13} & a_{12} & a_{11}
		\end{bmatrix}_{m \times n}.
		\]
	\end{dfn}
	
	Recall \cite{Brawley}, a matrix $A \in (R)_n,$ is called \emph{involutory} iff $A^2= I_n,$ where $I_n$ is the $n\times n$	identity matrix, $R$ is a finite commutative ring with identity and $(R)_n$ the set of $n\times n$ matrices over $R.$

	\begin{dfn}\cite[Definition 3]{Oztas}: Let $M = [a_{ij}]_{m\times n}$ and 
		$N = [b_{ij}]_{m\times n}$ be two $m \times n$ matrices over the ring $R$. 
		Then, the \emph{Hadamard product} of $M$ and $N$ is denoted by $M \circ N$ and is defined by
		\[
		M \circ N = [a_{ij}b_{ij}]_{m\times n},
		\]
		where $a_{ij}, b_{ij} \in R$ with $1 \le i \le m$, $1 \le j \le n$.
	\end{dfn}
	
	The Hadmard product of two matrices of the same order is simply the component wise matrix multiplication. Following \cite{Oztas}, we have the following properties of Hadmard product
	
	(i) The Hadamard $s$-power of $M = [a_{ij}]_{m\times n}$ is denoted by 
	$M^{\circ s}$, and defined as
	$
	M^{\circ s} = [(a_{ij})^s]_{m\times n}.
	$
	
	(ii) Since $R$ is a finite commutative ring, so	$M \circ N = N \circ M.$
	
	In the following, we state some results related to a special involutory permutation matrix and provide a representation of the double reflected matrix $M''$ of $M$ in terms of a suitable product of matrices.
	
	\begin{lem}\label{product of matrices}
		Let $R$ be a finite commutative ring with unity. For positive integers $k$ and $n$, let $J_n$ be the $n \times n$ permutation matrix with ones on the anti-diagonal and zeros elsewhere and $M$ be a $k\times n$ matrix over $R.$ Then
		\begin{enumerate}[(1)]
			\item $J_n^2$ is the $n\times n$ identity matrix over $R$, so, $J_n$ is an involutory matrix over $R$.
			\item For any $c=(c_0,\: c_1,...\:,c_{n-1})\in R^n$, the reverse of $c$ is $c^r=cJ_n$, where $c^r=(c_{n-1},\: c_{n-2},\:...,c_0)$.
			\item $J_kM=M_1$, where $M_1$ is the matrix obtained by reversing the order of the rows of $M$.
			\item $MJ_n=M_2$, where $M_2$ is the matrix obtained by reversing order of the columns of $M$.
			\item $J_k M J_n=M'',$ where $M''$ obtained by reversing the order of rows and columns of $M$ simultaneously, is the double reflected matrix of $M$.
			\item $[A\mid B]J_{2n}=[BJ_n\mid AJ_n]$ for any $k\times n$ matrices $A$ and $B$.
			\item $A^{os}J_n=(AJ_n)^{os}$, for any $k\times n$ matrix $A$, where $A^{os}$ denotes the Hadamard $s$ power of the matrix $A$.
		\end{enumerate}
	\end{lem}
	\begin{proof}
		The statements follow from direct computation of the matrix products involved.
	\end{proof}
	Next, we recall standard notions related to distance. Following \cite[Definition 2.3.1]{Xing}, if $x$ and $y$ are words of length $n$ over $R$ then the \emph{Hamming distance} from $x$ to $y$ denoted by $d(x,y)$ is defined to be the number of places at which $x$ and $y$ differ. The \emph{minimum distance} of a code $C$ over $R$, denoted by $d(C)$, is defined as $ d(C) = \min \{ d(x,y): x, y \in C \}.$ 
	Following \cite{Xing}, with respect to a given minimum Hamming distance $d$ for an $n$ length block code over an alphabet $A_q$ of size $q$, the \emph{Singleton bound} on the maximum size of the code, $A_q(n,d)$, is given as:
	\begin{propo}\cite[Theorem 5.4.1]{Xing}	For any integer $q \ge 1$, any positive integer $n$ and any integer $d$ such that 
		$1 \le d \le n$, we have
		\[
		A_q(n,d) \le q^{\,n-d+1}.
		\]
	\end{propo}
	Consequently, we have the following definition.	
	\begin{dfn}\cite[Definition 5.4.3]{Xing}
		A linear code with parameters $[n,k,d]$ such that $k+d = n+1$ 
		is called a maximum distance separable (MDS) code.
	\end{dfn}
	\begin{dfn}\cite[Definition 1]{De Broer} The Singleton defect of an $[n,k,d]$ code $C$ is $s(C)=n-k+1-d.$
	\end{dfn}
	\begin{dfn}\cite[Definition 5]{De Broer} A code with Singleton defect $s(C)=1$ is almost MDS (AMDS for short).
	\end{dfn}
	
	\section{Reversible linear codes by concatenation of matrices}In this section, we introduce the construction of a reversible linear code using a generator matrix obtained by concatenating any full rank $k\times n$ matrix over a finite commutative ring with unity and its $D-$reflected matrix, where $D$ is an arbitrary $k\times k$ involutory matrix. First, we define $D-$reflected matrix.
	\begin{dfn}\label{reflected}\label{D-reflection definition}
		Let $M$ be a $k\times n$ matrix over a ring $R$ and $D$ be any involutory $k\times k$ matrix over $R$, that is, $D\in R^{k\times k}$ such that $D^2=I_d.$ The $D-$ reflected matrix of $M$ denoted by $M'_D$ is defined as the matrix product $DMJ_n$. 
	\end{dfn}
	\begin{rmk}
		\begin{enumerate}[(1)]
			\item The $D-$ reflected matrix $M_D'=DMJ_n$ of $M$ is simply the product of the involutory matrix $D$ with the matrix $MJ_n$ obtained by reversing the order of the columns of $M$.
			\item By Lemma \ref{product of matrices} (5), $M''=J_kMJ_n.$ In particular if $D=J_k$, the notion of $D-$ reflected matrix $M_D'$ of $M$ is same as the notion of double reflected matrix $M''.$ 
			\item It is clear that for a given matrix $M$, the double reflected matrix $M''$ is a fixed matrix. However, in general we can construct a large number of $D-$reflected matrices of $M$ with respect to any involutory matrix $D$ over ring $R.$  
			\item There always exists a $D-$reflected matrix $M_D'$ for $D=J_k$ and in this case $M_D'=M''.$ However, the notion of $D-$ reflected matrices is more general, for example, over $R=\mathbb{F}_{16}$, there are 256 distinct $2\times2$ involutory matrices.
		\end{enumerate}
	\end{rmk}
	
	We give the following result showing that the code parameters of the linear codes generated by any full rank matrix $M$ over a finite commuatative ring with identity and $M_D'$  are identical. By a $k\times n$ matrix $M$ over a finite commutative ring with unity having $rank\:\:k$, we mean that the $k$ rows of $M$ are linearly independent over $R$. 
	\begin{thm}\label{parameters}Let $M$ be a $k\times n$ matrix over a finite commutative ring $R$ with unity such that $rank(M)=k$, $D$ be any $k\times k$ involutory matrix over $R$ and $M_D'$ be the associated $D-$ reflected matrix. Then the linear codes generated by $M$ and $M_D'$ have the same code parameters $[n,k,d].$ 
	\end{thm}
	\begin{proof} Let the linear code generated by $M$ over $R$ have parameters $n,k$ and $d$, where $n$ is the length of the code, $|R^k|$ the size of the code and $d$ the minimum Hamming distance of the code. Let $x$ be a codeword in the linear code generated by $M_D'$. Then by definition of a generator matrix, $x=aM_D'$ for some $a\in R^k$. Since $M_D'$ has order $k\times n$, $x=aM_D'$ has length $n$. Thus, the linear code generated by $M_D'$ has length $n$. Further, since $D$ is an involutory matrix over $R$, therefore $aD=0$ for some $a\in R^k$ implies $a=0$, for if $a\neq 0$ then $aD=0$ implies $aD^2=0$, a contradiction. We define a map $\bar{D}: R^k\longrightarrow R^k$ such that $\bar{D}(a)=a.D$ for all $a\in R^k.$ Let $\bar{D}(a)=\bar{D}(b)$ for some $a,b\in R^k$. Then $aD=bD$. Hence, $aD-bD=(a-b)D=0$. Thus, $a-b=0$. Therefore, $a=b$ because $D$ is involutory. Hence, $\bar{D}$ is bijective (R-isomomorphism) from $R^k$ to $R^k$. Thus $R^k=\{aD:\:a\in R^k\}$. Also, it is easy to verify that the $k$ rows of $MJ_n$ are linearly independent as $rank(M)=k$. Hence, the size of the code generated by $M_D'$ is $|R^k|$. Finally, since $R^k=\{aD:\:a\in R^k\}$, therefore the linear code generated by $M$, that is, $\{aM\: :a\in R^k\}$ is equal to $\{aDM\: :a\in R^k\}$. By Lemma \ref{product of matrices} (2), the reverse of $aDM$ is given by $aDMJ_n$. Since $aDMJ_n=aM_D'.$ Therefore the linear code generated by $M_D',$ that is, $\{aM_D'\: :a\in R^k\}$  is precisely the set of reverse of all the codewords of the linear code generated by $M$. Since reversal does not change the Hamming weight of a word, therefore the linear code generated by $M_D'$ has the same minimum non zero Hamming weight $d.$ Hence $M_D'$ has the same Hamming distance parameter, as the linear code generated by $M$.		
	\end{proof}
	As a consequence of above result, in particular for $D=J_k$, we have the following result. In rest of the generalization of results from \cite{Oztas}, here we take $D=J_k$ in proof.
	\begin{coro}\cite[Theorem 1]{Oztas} Let $M$ be a $k\times n$ matrix over the ring $R$ with $rank(M)=k$. Then, the codes $C=<M>$ and $C'=<M''>$ have the same $[n,k,d]$ parameters, where $<M>$ denotes the linear code generated by $M$.		
	\end{coro}
	Next, we show that the concatenation of $M$ with $M_D'$ results in a generator matrix which generates a reversible linear code over $R$.
	\begin{thm}\label{reversibility}Let $M$ be a $k\times n$ matrix over a finite commutative ring $R$ with unity such that $rank(M)=k$ and $D$ be any $k\times k$ involutory matrix over $R$. Then the linear code generated by the generator matrix $[M\mid M_D']$ is reversible.
	\end{thm} 
	\begin{proof}Let $x$ be a codeword in the linear code generated by $ [M\mid M_D']$ in $R^{2n}$. Then by definition of a generator matrix, $x=a[M\mid M_D']$ for some $a\in R^k.$ By Definition \ref{reflected}, $x=a[M\mid DMJ_n].$ Thus, $x=[aM\mid aDMJ_n].$ By Lemma \ref{product of matrices} $(2)$ and $(6)$, the reverse of $x$, $x^r=[aM\mid aDMJ_n]J_{2n}=[ aDMJ_nJ_n\mid aMJ_n].$ Using Lemma \ref{product of matrices} $(1)$, $x^r=[aDM\mid aMJ_n].$ Since $D^2=I_d$, therefore $x^r=[aDM\mid aD^2MJ_n].$ Hence, $x^r=aD[M\mid DMJ_n]$. Since $D$ is a map from $R^k$ to $R^k$, therefore $aD=b$ for some $b\in R^k$. Thus, $x^r=b[M\mid DMJ_n]=b[M\mid M_D'].$ So, $x^r$ belongs to the linear code generated by $[M\mid M_D']$. Hence, the linear code generated by $[M\mid M_D']$ is reversible.
	\end{proof}
	\begin{coro}\cite[Theroem 2]{Oztas} Let $M$ be a $k\times n$ matrix over the ring $R$ and $M''$ be double reflected matrix of $M$. Then, the code $<M\mid M''>$ is a reversible code over $R.$
	\end{coro}
	Next, we discuss the code parameters $n,k $ and $d$ of the linear code generated by the matrix $[M\mid M_D'].$
	\begin{thm}\label{reversible parameter} Let $M$ be a $k\times n$ matrix over a finite commutative ring $R$ with unity such that $rank(M)=k$ and $D$ be any $k\times k$ involutory matrix over $R$. Then the linear code generated by the generator matrix $[M\mid M_D']$ has length $2n$, size $|R^k|$ and distance greater than or equal to $2d$ where $n,k$ and $d$ are the code parameters of the linear code generated by $M$.  
	\end{thm}
	\begin{proof} Since each row of the concatenated matrix $[M\mid M_D']$ has length $2n$, each codeword being a linear combination of the rows of $[M\mid M_D']$ has length $2n$. Further, as $rank(M)$ is $k$, all the $k$ rows of the matrix $[M\mid M_D']$ are also linearly independent as lengthening of rows can not make two independent rows dependent. Thus, size of the code generated by $[M\mid M_D']$ is $|R^k|.$ Moreover, let $x$ be a codeword in the code generated by $[M\mid M_D']$. Then $x=(u\mid v)$ where $u$ and $v$ are codewords in the codes generated by $M$ and  $M_D',$ respectively. Since $d$ is the minimum distance of the code generated by $M$, therefore, the Hamming weight of $u$ is at least $d$. Also, by Theorem \ref{parameters}, $v$ has Hamming weight at least $d$. Thus $x$ has Hamming weight at least $2d$. Thus the minimum non zero Hamming weight which is equal to the minimum Hamming distance of the linear code generated by $[M\mid M_D']$ is at least $2d$.       
	\end{proof}
	
	\begin{coro}\cite[Theorem 3]{Oztas} Let $M$ be a $k\times n$ matrix over the ring $R$ with $rank(M) = k$. Then, the reversible code $C=<M\mid M''>$ has $[2n,k,d']$ parameters, where $d'\geq 2d$ and $<M>$ is a $[n, k, d]-$code.
	\end{coro} 
	Below we give two examples for the illustration of the construction of the concatenated generator matrix.
	\begin{exl} Consider the ring $R = \mathbb{Z}_4$ and a $2\times 4$ generator matrix $M =
		\begin{bmatrix}
			1 & 1 & 2 & 0 \\
			0 & 2 & 1 & 1
		\end{bmatrix}$. Then
		
		(1). We show that $M$ generates a $[4,2,2]$ linear code over $\mathbb{Z}_4.$ Clearly both the rows of $M$ are linearly independent. So $M$ generates a free module of rank $2$ over $\mathbb{Z}_4$. Also it is easy to verify that the distance of the linear code generated by $M$ is $2$.
		
		(2). Consider the $2 \times 2$ involutory matrix $D =
		\begin{bmatrix}
			1 & 2 \\
			0 & 1
		\end{bmatrix}$ over $R$ such that $D^2 =I_d$. Then $DM =
		\begin{bmatrix}
			1 & 1 & 0 & 2 \\
			0 & 2 & 1 & 1
		\end{bmatrix}$ and $M_D'= DMJ_4 =
		\begin{bmatrix}
			2 & 0 & 1 & 1 \\
			1 & 1 & 2 & 0
		\end{bmatrix}$. We concatenate the matrix $M$ with $M_D'$ to get a generator matrix $G =
		\begin{bmatrix}
			1 & 1 & 2 & 0 & 2 & 0 & 1 & 1 \\
			0 & 2 & 1 & 1 & 1 & 1 & 2 & 0
		\end{bmatrix}$.
		
		(3). By Theorem \ref{reversible parameter}, the linear code generated by $G$ over $\mathbb{Z}_4$ is reversible and has parameters $[8,2,d']$ where $d'\geq 4$. Direct computation shows $d'=4$. Thus $<G>$ is an $[8,2,4]$ reversible linear code over $\mathbb{Z}_4$.
	\end{exl}
	\begin{exl}\label{better example}
		Consider the finite field $\mathbb{F}_{16}$ and the $2\times 4$ matrix $M=\begin{bmatrix}
			1 & w & w^2 & w^3\\
			1 & w^2 & w^4 & w^6
		\end{bmatrix}$. It is easy to verify that $G$ is a rank 2 matrix over $\mathbb{F}_{16}.$
		
		(1). Consider the involutory matrix $J_2=\begin{bmatrix}
			0 & 1\\
			1 & 0			
		\end{bmatrix}$. Then the generator matrix obtained by concatenating $M$ with the double reflection $M''=J_2MJ_4$ of $M$, is given by 
		$G=\begin{bmatrix}
			1 & w & w^2 & w^3 & w^6 & w^4 & w^2 & 1\\
			1 & w^2 & w^4 & w^6 & w^3 & w^2 & w & 1
		\end{bmatrix}.$ By Theorem \ref{reversibility}, $G$ generates a reversible linear code. By SageMath computations \cite{Sage} the parameters of the code in this case are $[n,k,d]=[8,2,6]$. Thus $G$ generates an AMDS code. In the method of double reflected concatenation, we can construct only one generating matrix by which we have only AMDS code. 
		
		(2). If we take the involutory matrix $D=\begin{bmatrix}
			1 & 1\\
			0 & 1
		\end{bmatrix}$, then the concatenated generator matrix $G'=[M\mid M_D']=[M\mid DMJ_6]$ is given by $G'=\begin{bmatrix} 
			1 & w & w^2 & w^3 & w^2 & w^{10} & w^5 & 0\\
			1 & w^2 & w^4 & w^6 & w^6 & w^4 & w^2 & 1
		\end{bmatrix}.$ By Theorem \ref{reversibility}, the linear code generated by $G'$ is reversible. By SageMath computations \cite{Sage} the parameters of the code in this case are $[n,k,d]=[8,2,7]$. Hence, we get a linear reversible MDS code. In our method of $D-$reflected concatenation, we can construct not only AMDS code but also MDS code by suitable choices of involutory matrix over $R$ with respect to same $M$ over $\mathbb{F}_{16}.$
	\end{exl}
	\begin{rmk}\label{large number of mds} (1). Example \ref{better example} shows that by choosing an appropriate involutory matrix in place of double reflection, the code parameters can be improved while keeping the reversibility intact. It also highlights that we can even construct MDS reversible codes by adequate involutory matrices.
		
		(2). In Example \ref{better example}, the alphabet used is $\mathbb{F}_{16}$. By SageMath computation \cite{Sage}, we obtain 256 involutory matrices, $D$ such that $<M\mid M_D'>$ is a reversible linear code, out of which 120 are MDS codes and 136 are AMDS codes. For larger alphabets, we get even large number of involutory matrices which serve the purpose. Thus, we get a large class of generator matrices which generate different reversible linear codes.
		
	\end{rmk}
	Next, we develop a similar construction technique of generator matrices over the finite field $\mathbb{F}_{4^{2t}}$ where $t$ is a positive integer. We then show its application in the construction of reversible DNA codes.
	\section{Reversible self dual codes} In this section we discuss the construction of linear codes which are both self dual as well as reversible using the concatenation of matrices approach for a generator matrix. First we mention the following result for the blockwise product of matrices.
	\begin{lem}\cite[Theorem 1.3.1]{Golub}\label{block product} If
		\[
		A=
		\begin{array}{c}
			\left[
			\begin{array}{ccc}
				A_{11} & \cdots & A_{1s}\\
				\vdots &        & \vdots\\
				A_{q1} & \cdots & A_{qs}
			\end{array}
			\right]
			\\[-2mm]
			\begin{array}{ccc}
				p_1 && p_s
			\end{array}
		\end{array}
		\begin{array}{c}
			m_1\\
			\\
			m_q
		\end{array},
		\qquad
		B=
		\begin{array}{c}
			\left[
			\begin{array}{ccc}
				B_{11} & \cdots & B_{1r}\\
				\vdots &        & \vdots\\
				B_{s1} & \cdots & B_{sr}
			\end{array}
			\right]
			\\[-2mm]
			\begin{array}{ccc}
				n_1 && n_r
			\end{array}
		\end{array}
		\begin{array}{c}
			p_1\\
			\\
			p_s
		\end{array},
		\]
		and we partition the product \(C=AB\) as follows,
		\[
		C=
		\begin{array}{c}
			\left[
			\begin{array}{ccc}
				C_{11} & \cdots & C_{1r}\\
				\vdots &        & \vdots\\
				C_{q1} & \cdots & C_{qr}
			\end{array}
			\right]
			\\[-2mm]
			\begin{array}{ccc}
				n_1 && n_r
			\end{array}
		\end{array}
		\begin{array}{c}
			m_1\\
			\\
			m_q
		\end{array},
		\]
		then for \(\alpha=1\!:\!q\) and \(\beta=1\!:\!r\) we have
		\[
		C_{\alpha\beta}
		=
		\sum_{\gamma=1}^{s}
		A_{\alpha\gamma}B_{\gamma\beta}.
		\]

	\end{lem}
	\begin{lem}\label{self ortho ggt} \cite[Corollary 2.1 (1)]{Thipwo} Let $C$ be a linear code with generator matrix $G$, then
		$C$ is Euclidean self-orthogonal if and only if $GG^T = [0]$.
	\end{lem}
	\begin{lem}\label{n by 2}\cite[Theorem 1.5.2]{Huffman}
		\textit{Let $\mathcal{C}$ be an $[n,k]_q$ code with generator and parity check matrices $G$ and $H$, respectively. Then $\mathcal{C}^\perp$ is an $[n, n-k]_q$ code with generator and parity check matrices $H$ and $G$, respectively. Additionally $(\mathcal{C}^\perp)^\perp = \mathcal{C}$. Furthermore $\mathcal{C}$ is self-dual if and only if $\mathcal{C}$ is self-orthogonal and $k = \dfrac{n}{2}$.}
	\end{lem}

	\begin{lem}\label{matrix transpose}
		Let $A\in R^{m\times n}$ and $B\in R^{m\times p}$ be two matrices over a ring $R$. Then
		\[
		[A\mid B]^T=
		\begin{bmatrix}
			A^T\\
			B^T
		\end{bmatrix}.
		\]
	\end{lem}
	
	\begin{proof}
		Let $A=(a_{ij})$ and $B=(b_{ij})$. Then
		\[
		\begin{aligned}
			[A\mid B]^T
			&=
			\begin{bmatrix}
				a_{11}&\cdots&a_{1n}&b_{11}&\cdots&b_{1p}\\
				\vdots&&\vdots&\vdots&&\vdots\\
				a_{m1}&\cdots&a_{mn}&b_{m1}&\cdots&b_{mp}
			\end{bmatrix}^T\\[6pt]
			&=
			\begin{bmatrix}
				a_{11}&\cdots&a_{m1}\\
				\vdots&&\vdots\\
				a_{1n}&\cdots&a_{mn}\\
				b_{11}&\cdots&b_{m1}\\
				\vdots&&\vdots\\
				b_{1p}&\cdots&b_{mp}
			\end{bmatrix}\\[6pt]
			&=
			\begin{bmatrix}
				A^T\\
				B^T
			\end{bmatrix}.
		\end{aligned}
		\]
		This completes the proof.
	\end{proof}
	\begin{thm}\label{self dual reversible} Let $n$ be a positive even integer. Let $F_q$ be a finite field of order $q$ and $M$ be a $k\times n$ matrix over $F_q$ such that $rank(M)=k$ and $k=\frac{n}{2}$. Let $D$ be a $k\times k$ involutory matrix over $F_q$. Then the linear code generated by $[M\mid M_D']$ is a self dual reversible linear code if and only if the matrix $DMM^T$ is skew symmetric, where $M^T$ is the transpose of $M.$
	\end{thm}
	\begin{proof} Let $G=[M\mid M_D']$, and $C$ be the linear code generated by $G.$ By Theorem \ref{reversibility}, the linear code generated by $G$ is reversible. Thus, it suffices to show that $C$ is self dual if and only if $DMM^T$ is skew symmetric. We observe that $DMM^T$ is a square matrix of order $k\times k,$ where $k=\frac{n}{2}.$ Hence  
		
		\begin{align*}
			C \text{ is self-dual}
			&\iff GG^{T}=0
			\qquad
			\left(
			\begin{array}{l}
				 \text{ by Lemma}\; \ref{self ortho ggt} \; \text{and}\;\ref{n by 2}
			\end{array}
			\right) \\
			&\iff [M\,|\,M_D'][M\,|\,M_D']^{T}=0 \\
			&\iff [M\,|\,M_D']
			\begin{bmatrix}
				M^T\\
				M_D'^T
			\end{bmatrix}=0 
			\qquad
			\left(
			\begin{array}{l}
				\text{by Lemma }\; \ref{matrix transpose}
			\end{array}
			\right) \\
			&\iff MM^{T}+M_D M_D'^{T}=0
			\qquad (\text{by Lemma}\; \ref{block product}) \\
			&\iff MM^{T}+(DMJ_n)(DMJ_n)^{T}=0 \\
			&\iff MM^{T}+DMJ_nJ_n^{T}M^{T}D^{T}=0 \\
			&\iff MM^{T}+DMM^{T}D^{T}=0
			\qquad
			\left(
			\begin{array}{l}
				\text{since } J_n^{T}=J_n\\
				\text{and } J_n^{2}=I_d
			\end{array}
			\right) \\
			&\iff DMM^{T}+MM^{T}D^{T}=0
			\qquad
			\left(
			\begin{array}{l}
				\text{multiplying both sides}\\
				\text{by } D \text{ and using } D^{2}=I_d
			\end{array}
			\right) \\
			&\iff DMM^{T}+(DMM^{T})^{T}=0 \\
			&\iff DMM^{T}=-(DMM^{T})^{T} \\
			&\iff DMM^{T}
			\text{ is a skew-symmetric matrix.}
		\end{align*}		
	\end{proof} 
	
	\begin{coro}\label{self duality preserved} Let $n$ be a positive even integer. Let $F_q$ be a finite field of order $q$ and $M$ be a $k\times n$ matrix over $F_q$ such that $rank(M)=k$ and $k=\frac{n}{2}$. Let $D$ be a $k\times k$ involutory matrix over $F_q$.
    If the linear code generated by $M$ is self dual then the linear code generated by $[M\mid M_D']$ is a self dual reversible linear code.
	\end{coro}
	\begin{proof} By Theorem \ref{reversibility} the linear code generated by $[M\mid M_D']$ is reversible. Further, since the linear code generated by $M$ is self dual, therefore by Lemma \ref{self ortho ggt}, $MM^T=0.$ This implies $DMM^T=0$. Since, null matrix is skew symmetric, therefore by Theorem \ref{self dual reversible}, the linear code generated by $[M\mid M_D']$ is self dual.
		
	\end{proof} 
	\begin{exl}
		Consider the binary field $\mathbb{Z}_2$. There are precisely four $2\times 2$ involutory matrices over $\mathbb{Z}_2,$ given by $
		D_1=
		\begin{pmatrix}
			1 & 0\\
			0 & 1
		\end{pmatrix}$, $ D_2=
		\begin{pmatrix}
			1 & 1\\
			0 & 1
		\end{pmatrix}
		$, $D_3=
		\begin{pmatrix}
			1 & 0\\
			1 & 1
		\end{pmatrix}
		$ and
		$D_4=
		\begin{pmatrix}
			0 & 1\\
			1 & 0
		\end{pmatrix}.
		$ Consider the $2\times 4$ binary matrix $M=\begin{pmatrix}
			1 & 1 & 1 & 1\\
			1 & 0 & 1 & 0
		\end{pmatrix}.$ Clearly the rank of $M$ is 2. Further it is easy to check that $MM^T=0$. Hence $M$ generates a binary self dual code. By Corollary \ref{self duality preserved}, the linear code generated by $[M\mid M_{D_i}']$ is self dual as well as reversible for each $i=1,2,3,4.$
	\end{exl}
	
	\section{Reversible DNA codes using concatenation of matrices}
	The DNA codewords over the alphabet $\{A,T,G,C \}$ are devoid of any inherent algebraic structure. Hence, we use a DNA correspondence map from an appropriate algebraic alphabet to DNA bases to apply the theory of algebraic coding to DNA coding. We utilize the DNA correspondence table for $\mathbb{F}_{16}$ developed by Oztas and Siap \cite{Siap}, also mentioned in \cite[Table 2]{Oztas}.

	\begin{table}[h]\label{table 1}
		\centering
		\caption{DNA correspondence over $\mathbb{F}_{16},$ \cite[Table 2]{Oztas}} \label{DNA table}
		\begin{tabular}{|c|c|c|c|} \hline
			
			Index & DNA Pair & $\mathbb{F}_{16}$ (Multiplicative ) & Additive Form \\ \hline
			
			1  & AA & $0$       & $0$ \\
			2  & TT & $w^0$     & $1$ \\
			3  & AT & $w^1$     & $w$ \\
			4  & GC & $w^2$     & $w^2$ \\
			5  & AG & $w^3$     & $w^3$ \\
			6  & TA & $w^4$     & $1+w$ \\
			7  & CC & $w^5$     & $w+w^2$ \\
			8  & AC & $w^6$     & $w^2+w^3$ \\
			9  & GT & $w^7$     & $1+w+w^3$ \\
			10 & CG & $w^8$     & $1+w^2$ \\
			11 & CA & $w^9$     & $w+w^3$ \\
			12 & GG & $w^{10}$  & $1+w+w^2$ \\
			13 & CT & $w^{11}$  & $w+w^2+w^3$ \\
			14 & GA & $w^{12}$  & $1+w+w^2+w^3$ \\
			15 & TG & $w^{13}$  & $1+w^2+w^3$ \\
			16 & TC & $w^{14}$  & $1+w^3$ \\ \hline
			
		\end{tabular}
	\end{table}
	We recall, remarks from \cite[Section IV]{Oztas} as follows: let $\Phi$ represent the bijective mapping that defines DNA correspondence for the elements of the field. For instance,
	$\Phi(w^{11}) = CT$ over $\mathbb{F}_{16}$. If we add $1$ to any element of
	$\mathbb{F}_{16}$, then its image under $\Phi$ produces a complement of
	the DNA correspondence. For example,
	$\Phi(w^{11} + 1) = \Phi(w^{12}) = GA$. Additionally, taking the
	$4^{t}$-th power of the elements produces the reverse of the DNA correspondence. For
	instance, $\Phi((w^{11})^{4}) = \Phi(w^{14}) = TC$. This map
	$\Phi$ can be extended to codewords and codes.
	\begin{rmk}
		In \cite{Oztas}, authors raised an open problem: " Let $C$ be any non-reversible linear code over the field $F_{4^{2t}}$. Then, $\Phi(C)$ is a reversible DNA code, where $\Phi$ denotes the DNA correspondence map." We observe that the answer to the open problem is negative. In support, we give the following example.
	\end{rmk}
	
	\begin{exl}\label{open problem example}
		Consider the finite field $F_{4^{2t}}=\mathbb{F}_{16}$, that is, $t=1$ and the linear code
		\[
		C = \{(x,0) : x \in \mathbb{F}_{16}\} \subseteq \mathbb{F}_{16}^2.
		\]
		Clearly, $C$ is a $1$-dimensional linear code over $\mathbb{F}_{16}$.
		
		We first show that $C$ is not reversible. Recall that a code $C$ is said to be reversible if for every $(c_1,c_2) \in C$, the reversed vector $(c_2,c_1)$ also belongs to $C$. Now, take any nonzero element $x \in \mathbb{F}_{16}$. Then $(x,0) \in C$, but its reversal is $(0,x),$ which is not of the form $(y,0)$ for any $y \in \mathbb{F}_{16}$ unless $x=0$. Hence, $(0,x) \notin C$ for $x \neq 0$, and therefore $C$ is not reversible.
		
		Now consider the DNA correspondence map $\Phi : \mathbb{F}_{16} \to \{\text{A,T,G,C}\}^2$ defined as in \cite[Table 2]{Oztas}. Extending $\Phi$ coordinate-wise, we obtain
		$
		\Phi(C) = \{(\Phi(x), \Phi(0)) : x \in \mathbb{F}_{16}\}.
		$
		From the correspondence table, we have $\Phi(0) = \mathrm{AA}$. Hence,
		$
		\Phi(C) = \{(\Phi(x), \mathrm{AA}) : x \in \mathbb{F}_{16}\}.
		$
		
		We now check whether $\Phi(C)$ is reversible. Take any element $(\Phi(x), \mathrm{AA}) \in \Phi(C)$. Its reversal is
		$
		(\mathrm{AA}, \Phi(x)^r)
		$, where $\Phi(x)^r$ is the reverse of $\Phi(x).$
		For $\Phi(C)$ to be reversible, this reversed vector must also belong to $\Phi(C)$. That is, there must exist some $y \in \mathbb{F}_{16}$ such that
		$
		(\mathrm{AA}, \Phi(x)^r) = (\Phi(y), \mathrm{AA}).
		$
		Comparing coordinates, we obtain
		$
		\Phi(y) = \mathrm{AA} \quad \text{and} \quad \Phi(x)^r = \mathrm{AA}.
		$
		Since $\Phi$ is a bijection, this implies $y = 0$ and $x = 0$. Therefore, the reversed vector belongs to $\Phi(C)$ only in the trivial case.
		
		Thus, for any $x \neq 0$, we have
		$
		(\mathrm{AA}, \Phi(x)) \notin \Phi(C),
		$
		which shows that $\Phi(C)$ is not reversible.
	\end{exl}
	\begin{thm}\label{same parameters DNA}
		Let $M$ be a $k\times n$ matrix over a finite field $\mathbb{F}_{2^r}$ such that $rank(M)=k$ and $D$ be any $k\times k$ involutory matrix over $\mathbb{F}_{2^r}$. Then the linear codes generated by $M$ and $(M_D')^{o2^t}$ have the same code parameters $[n,k,d],$ where $2^t< 2^r-1$. 
	\end{thm}
	\begin{proof} Since $M$ and $(M_D')^{o2^t}$ have the same order $k\times n$, the linear codes generated by $M$ and $(M_D')^{o2^t}$ have the same length $n$. Further, it is a standard result that the linear independence of vectors is preserved by a Frobenius automorphism. Since, $2^t< 2^r-1$, the map $x\longrightarrow x^{2^t}$ is a Frobenius automorphism of $F_{2^r}$. Therefore the linear codes generated by $M$ and $(M_D')^{o2^t}$ have the same dimension $k$.
		
		Next we show that both the codes have same distance parameter. By Theorem \ref{parameters}, by putting $R=\mathbb{F}_{2^r}$, the linear codes generated by $M$ and $M_D'$ have the same minimum distance. Let $g^1,g^2,...,g^k$ be the $k$ rows of the matrix $M_D'$. Then the rows of the matrix $(M_D')^{o2^t}$ are $(g^1)^{o2^t},(g^2)^{o2^t},...,(g^k)^{o2^t}$, where $(g^i)^{o2^t}=((g^i_j)^{2^t})_n$. Let $c=\sum_{i=1}^{k}\alpha_ig^i$ be a codeword of weight $w$ in the linear code generated by $M_D'$. For arbitrary $i,l$ and $j$ such that $1\leq i,l\leq k$ and $1\leq j\leq n$ consider the entries $g^i_j$ and $g^l_j$ of the matrix $M_D'$. If $\alpha g^i_j + \beta g^l_j\neq 0$ for some scalars $\alpha$ and $\beta$ then $(\alpha g^i_j)^{2^t} + (\beta g^l_j)^{2^t}=(\alpha g^i_j + \beta g^l_j)^{2^t}\neq 0$, as the characterstic of the field $\mathbb{F}_{2^r}$ is $2$. 
		
		Conversely, if $\alpha g^i_j + \beta g^l_j= 0$ then $(\alpha g^i_j)^{2^t} + (\beta g^l_j)^{2^t}=(\alpha g^i_j + \beta g^l_j)^{2^t}=0$. Now consider the scalars $(\alpha_i)^{2^t}$. Then $c'=\sum_{i=1}^{k}(\alpha_i)^{2^t}(g^i)^{2^t}$ being a linear combination of the rows of the matrix $(M_D')^{o2^t}$ is a codeword of the linear code generated by $(M_D')^{o2^t}$. Also it has the same number of non zero components as $c$. Thus for every codeword $c$ of weight $w$ in the linear code generated by $M_D'$ there corresponds a codeword $c'$ of weight $w$ in the linear code generated by $(M_D')^{o2^t}$. Since both the linear codes have the same size, that is, parameter $k$, therefore they have exactly the same weight distribution. Hence, the minimum Hamming distance for both the codes is same. 
	\end{proof}
	
	\begin{coro}\cite[Theorem 4]{Oztas}
		Let $M$ be a $k \times n$ matrix over $F_{2^{r}}$ with
		$rank(M)=k$. Then, the codes $\langle M \rangle$ and
		$\langle (M'')^{o2^{t}} \rangle$ have same $[n,k,d]$
		parameters, where $2^{t} < 2^{r}-1$.
	\end{coro}
	
	\begin{coro}\label{same parameter DNA coro}Let $M$ be a $k\times n$ matrix over a finite field $\mathbb{F}_{4^{2t}}$ such that $rank(M)=k$ and $D$ be any $k\times k$ involutory matrix over $\mathbb{F}_{4^{2t}}$. Then the linear codes generated by $M$ and $(M_D')^{o4^t}$ have the same code parameters $[n,k,d].$
	\end{coro}
	\begin{coro}\cite[Corollary 1]{Oztas} Let $M$ be a $k \times n$ matrix over $F_{4^{2t}}$ with $rank(M)=k$. Then, the codes $\langle M \rangle$ and	$\langle (M'')^{o^{4^{t}}} \rangle$ have same $[n,k,d]$ parameters.
	\end{coro}
	Next, we give the main result which shows the application of concatenation of appropriate matrices to generate a linear code over $F_{4^{2t}}$ such that its $\Phi$ image is a reversible DNA code. First, we give the following results highlighting the relation between Hadmard power and frobenius automorphisms. Also, Watson-Crick complement of any $\Phi(\alpha)$ in DNA pairs.
	
	In general, Hadamard power does not distribute over matrices product, that is, $(AB)^{os}\neq A^{os}B^{os}.$ For example, let
	$
	A=
	\begin{pmatrix}
		1 & 2\\
		0 & 1
	\end{pmatrix}
	$ and
	$B=
	\begin{pmatrix}
		2 & 1\\
		3 & 0
	\end{pmatrix},
	$ be two real $2 \times 2$ matrices. Then $ A^{\circ 2}
	=
	\begin{pmatrix}
		1 & 4\\
		0 & 1
	\end{pmatrix}$ and $B^{\circ 2}
	=
	\begin{pmatrix}
		4 & 1\\
		9 & 0
	\end{pmatrix}.$
	Now,
	$
	AB=
	\begin{pmatrix}
		8 & 1\\
		3 & 0
	\end{pmatrix}.
	$
	Hence
	$
	(AB)^{\circ 2}
	=
	\begin{pmatrix}
		64 & 1\\
		9 & 0
	\end{pmatrix}.
	$
	On the other hand,
	$
	A^{\circ 2}B^{\circ 2}
	=
	\begin{pmatrix}
		1 & 4\\
		0 & 1
	\end{pmatrix}
	\begin{pmatrix}
		4 & 1\\
		9 & 0
	\end{pmatrix}
	=
	\begin{pmatrix}
		40 & 1\\
		9 & 0
	\end{pmatrix}.
	$
	Therefore,
	$
	(AB)^{\circ 2}
	\neq
	A^{\circ 2}B^{\circ 2}.
	$ 
	
	However, in case of Frobenius (Hadamard) power it holds.

	\begin{lem}\label{frobenius}
		Let $\mathbb{F}_q$ be a finite field of characteristic $p$, where $q=p^m$, and let $t$ be a nonnegative integer. For matrices $A \in M_{m \times k}(\mathbb{F}_q)$ and $B \in M_{k \times n}(\mathbb{F}_q)$, consider the Frobenius power $A^{(p^t)}$ by applying the map $x \mapsto x^{p^t}$ on $\mathbb{F}_q$ on each entry of $A.$ Then $(AB)^{(p^t)} = A^{(p^t)} B^{(p^t)}$.
	\end{lem}
	
	\begin{proof}
		We begin by recalling the Frobenius map defined by $x\mapsto x^{p^t}$ which is a field automorphism. In particular, for all $x,y \in \mathbb{F}_q$, it satisfies $(x+y)^{p^t} = x^{p^t} + y^{p^t}$ and $(xy)^{p^t} = x^{p^t}y^{p^t}$. 
		
		Let $A=(a_{ir}) \in M_{m \times k}(\mathbb{F}_q)$ and $B=(b_{rj}) \in M_{k \times n}(\mathbb{F}_q)$. Then the $(i,j)$-th entry of the product $AB$ is given by $(AB)_{ij} = \sum_{r=1}^{k} a_{ir} b_{rj}$. Applying the Frobenius map on each entry of $AB$, we obtain $(AB)_{ij}^{p^t} = \left(\sum_{r=1}^{k} a_{ir} b_{rj}\right)^{p^t}.$
		
		Using the additivity of the Frobenius map, this becomes $(AB)_{ij}^{p^t} = \sum_{r=1}^{k} (a_{ir} b_{rj})^{p^t}$. Further, by multiplicativity, we have $(AB)_{ij}^{p^t} = \sum_{r=1}^{k} a_{ir}^{p^t} b_{rj}^{p^t}$. 
		
		On the other hand, consider the matrices $A^{(p^t)}=(a_{ir}^{p^t})$ and $B^{(p^t)}=(b_{rj}^{p^t})$. Their product has $(i,j)$-th entry $(A^{(p^t)}B^{(p^t)})_{ij} = \sum_{r=1}^{k} a_{ir}^{p^t} b_{rj}^{p^t}$. Comparing the two expressions, we see that $(AB)_{ij}^{p^t} = (A^{(p^t)}B^{(p^t)})_{ij}$ for all $i,j$. Hence $(AB)^{(p^t)} = A^{(p^t)} B^{(p^t)}.$
	\end{proof}
	We restate the following lemma given in \cite{Oztas}.
	
	\begin{lem}\cite[Lemma 1]{Oztas}\label{reverse and complement}: For any $\alpha \in F_{4^{2t}}$, $\Phi(\alpha^{4^t})=\Phi(\alpha)^r$ and $\Phi(\alpha+1)=\Phi(\alpha)^c$, where $\Phi(\alpha)^r$ is the reverse of $\Phi(\alpha)$ and $\Phi(\alpha)^c$ is the Watson-Crick complement.
	\end{lem}			
	\begin{thm}\label{main theorem}Let $M$ be a $k\times n$ matrix over a field $F_{4^{2t}}$ such that $rank(M)=k$ and $D$ be any $k\times k$ involutory matrix over $F_{4^{2t}}$ such that $D^{o4^t}=D.$ Let $C$ be the linear code generated by $[M\mid (M'_D)^{o4^t}]$. Then the code $\Phi(C)$ is reversible. Moreover, the code parameters of $C$ are $2n,\:k,\:d'$ where $d'\geq 2d$ and $n,k,d$ are the code parameters of the linear code generated by $M.$
	\end{thm}
	\begin{proof}
		Let $G=[M\mid (M_D')^{o4^t}]$ generate the linear code $C.$ Let $c\in C$. Then by definition of a generator matrix, there exists $a\in F_{4^{2t}}^k$ such that $c =aG$. Thus the DNA codeword corresponding to $c$ is given by $\Phi(aG).$ Hence, $\Phi(c)=\Phi(a[M\mid (M_D')^{o4^t}])=\Phi([aM\mid a(M_D')^{o4^t}]).$ Since $M_D'=DMJ_n$, therefore $\Phi(c)=\Phi([aM\mid a(DMJ_n)^{o4^t}]).$ As $D^{o4^t}=D$ so by Lemma \ref{frobenius}, $(DMJ_n)^{o4^t}=D^{o4^t}M^{o4^t}J_n^{o4^t}=DM^{o4^t}J_n$. Hence, $\Phi(c)= \Phi([aM\mid aDM^{o4^t}J_n])$. By Lemma \ref{product of matrices} (2) and Lemma \ref{reverse and complement}, the reverse of $\Phi(c)$ is given by $\Phi(c)^r=\Phi([(aDM^{o4^t}J_n)^{o4^t}J_n\mid (aM)^{o4^t}J_n]).$ Hence, using Lemma \ref{product of matrices} (7), $\Phi(c)^r=\Phi([(aDM^{o4^t}J_nJ_n)^{o4^t}\mid (aMJ_n)^{o4^t}]).$ Since $J_n^2=D^2=I_d$, therefore $\Phi(c)^r=\Phi([(aDM^{o4^t})^{o4^t}\mid (aDDMJ_n)^{o4^t}]).$ Using Lemma \ref{frobenius}, $\Phi(c)^r=\Phi([(aD)^{o4^t}(M^{o4^t})^{o4^t}\mid (aD)^{o4^t}(DMJ_n)^{o4^t}])$. Putting $(aD)^{o4^t}=b\in F_{4^{2t}}^k$, we get, $\Phi(c)^r=\Phi([bM^{o4^{2t}}\mid b(M_D')^{o4^t}]).$ Since $\alpha^{4^{2t}}=\alpha $ for every $\alpha$ in $F_{4^{2t}}$, therefore $M^{o4^{2t}}=M$. Hence, $\Phi(c)^r=\Phi([(bM)\mid b(M_D')^{o4^t}])=\Phi(b[M\mid (M_D')^{o4^t}])=\Phi(bG).$ Hence, $\Phi(c)^r=\Phi(c')$, where $c'=bG$ belongs to the linear code $C$. Thus, $\Phi(c)^r\in \Phi(C).$ Therefore, $\Phi(C)$ is a reversible DNA code.
		
		Further, as order of the matrix $[M\mid (M'_D)^{o4^t}]$ is $k\times 2n$, so the length of the code $C$ is $2n$. Also, since $rank(M)=k$ and lengthening of the rows of a matrix cannot decrease the number of linearly independent rows, therefore $rank([M\mid (M'_D)^{o4^t}])=k$. Hence, the dimension of $C$ is $k$. Finally, if $u$ is a codeword in $C$, then $u=aG$ for some $a\in F_{4^{2t}}$. Thus, $u=a[M\mid (M'_D)^{o4^t}]=[aM\mid a(M'_D)^{o4^t}]=[c\mid c']$, where $c=aM$ and $c'=a(M'_D)^{o4^t}$. Using Corollary \ref{same parameter DNA coro}, both $c$ and $c'$ have Hamming weight greater than or equal to $d$. Thus the Hamming weight of $u$ being equal to the sum of weights of $c$ and $c'$ is at least $2d.$ Since, $u$ is an arbitrary codeword of $C$, therefore the minimum Hamming weight parameter of $C$ is $d'$ such that $d'\geq 2d.$ 
	\end{proof}
	\begin{coro} \cite[Theorem 5]{Oztas} Let $C'=<M>$ and $C''=<(M'')^{o4^t}>$ be two $[n,k,d]$-code over $\mathbb{F}_{4^{2t}}.$ Then, the code $C=<M\mid (M'')^{o4^t}>$ is a $[2n,k,d']-$linear code over $\mathbb{F}_{4^{2t}}$, where $d'\geq 2d$, and $\Phi(C)$ corresponds to a reversible DNA code.
	\end{coro}
	\begin{rmk}\cite[Corollary 3]{Oztas}\label{complement} states: "let $C'=<M>$ be an $[n,k,d]$-linear code over $F_{4^{2t}}$ and $(1,1,1,...,1)\in C'$. Then $\Phi(C)$ is a reversible complement DNA code, where $C=<M\mid (M'')^{o4^t}>$ over $F_{4^{2t}}$."
		
		In the following, we show that this result is not true in general. In support, we provide a characterization of a DNA code closed under complementation and a counter example. Further in Proposition \ref{right characterisation}, we provide the correct version of that result. 
	\end{rmk}		    
	\begin{propo}\label{characterisation of complementation} Let $C$ be a non-zero linear code over $F_{4^{2t}}$. If $\Phi(C)$ is the DNA code corresponding to $C$, where $\Phi$ is as given in \cite[Table 2]{Oztas}
		then $\Phi(C)$ is closed under complementation if and only if
		$(1,1,\ldots,1) \in C$.
	\end{propo}

	\begin{proof} Let $(1,1,\ldots,1) \in C$. Since, $C$ is a linear code, for any $x$ in $C$, $x + (1,1,\ldots,1) \in C.$ It follows that the DNA complement of the DNA codeword $\Phi(x)$ is $\overline{\Phi(x)} = \Phi\big(x + (1,1,\ldots,1)\big)$ in $\Phi(C).$  Hence, $\Phi(C)$ is closed under complementation.
		
		\medskip
		
		Conversely, let $x \in C$. Then 
		$
		\Phi(x) \in \Phi(C).
		$
		Since $\Phi(C)$ is closed under complementation, therefore
		$
		\overline{\Phi(x)} \in \Phi(C).
		$
		This implies that
		$
		\Phi\big(x + (1,1,\ldots,1)\big) \in \Phi(C).
		$
		Thus 
		$
		x + (1,1,\ldots,1) \in C.
		$   
		Since $C$ is linear, therefore $(x + (1,1,\ldots,1)) - x = (1,1,\ldots,1) \in C.
		$ 
	\end{proof}
	Now we give a counter example to \cite[Corollary 3]{Oztas}.
	
	\begin{exl}\label{wrong characterization example}
		Consider $\mathbb{F}_{4^{2t}} = \mathbb{F}_{16}$, that is, $t=1$. Let 
		$M = \begin{bmatrix} 1 & 1 & 1 \\ 0 & 1 & 0 \end{bmatrix}$. 
		Then $M'' = \begin{bmatrix} 0 & 1 & 0 \\ 1 & 1 & 1 \end{bmatrix}$. 
		Thus, $(M'')^{o4} = \begin{bmatrix} 0 & 1 & 0 \\ 1 & 1 & 1 \end{bmatrix}$. 
		Clearly, both $M$ and $(M'')^{o4}$ have rank $2$ over $\mathbb{F}_{16}$ and so they induce injective maps from $\mathbb{F}_{16}^2$ to $\mathbb{F}_{16}^3$. 
		By Proposition \ref{main theorem}, the linear code $C$ generated by 
		$[\, M \mid (M'')^{o4} \,]$ corresponds to a reversible DNA code $\Phi(C)$.
		
		Consider $a = [1\ 0] \in \mathbb{F}_{16}^2$. Then 
		$aM = [1\ 0]\begin{bmatrix} 1 & 1 & 1 \\ 0 & 1 & 0 \end{bmatrix} = [1\ 1\ 1]$. 
		Thus $[1\ 1\ 1]$ belongs to the linear code generated by $M$. Hence, by \cite[Corollary 3]{Oztas} , 
		$\Phi(C)$ should be a reversible complement DNA code. However, since $M$ is an injective map, therefore  	$aM = [1\ 1\ 1]$ implies $bM \neq [1\ 1\ 1]$ for all $b \neq a$ in $\mathbb{F}_{16}^2$. 
		Also, $a[\, M \mid (M'')^{o4} \,] =a[\, M \mid M'' \,]= [\, aM \mid aM'' \,] \neq [1\ 1\ 1\ 1\ 1\ 1]$, 
		because $aM'' = [1\ 0]\begin{bmatrix} 1 & 0 & 1 \\ 0 & 0 & 1 \end{bmatrix} = [0\ 1\ 0]$. 
		Therefore, $a[\, M \mid M'' \,] \neq [1\ 1\ 1\ 1\ 1\ 1]$ for any $a \in \mathbb{F}_{16}^2$. 
		It follows that $[1\ 1\ 1\ 1\ 1\ 1]$ vector does not belong to $C$. Hence, by Proposition \ref{characterisation of complementation}, 
		$\Phi(C)$ is not closed under complementation.
	\end{exl}
	\begin{propo}\label{right characterisation}
		Let $M$ be a $k\times n$ matrix over a field $F_{4^{2t}}$ such that $rank(M)=k$ and $D$ be any $k\times k$ involutory matrix over $F_{4^{2t}}$ such that $D^{o4^t}=D.$ Let $C$ be the linear code generated by $[M\mid (M'_D)^{o4^t}]$. Then the code $\Phi(C)$ is reversible. Further if $(1,1,1,...,1)\in C$ then $\Phi(C)$ is a reversible complement DNA code.
		
	\end{propo}
	\begin{proof}
		It follows from Proposition \ref{characterisation of complementation} and Theorem \ref{main theorem}.
	\end{proof}
	
	We give an example of a reversible DNA code constructed as an image of a linear code over $\mathbb{F}_{16}$ generated by a matrix which is obtained by concatenating a given matrix $M$ with its $D-$reflection, $M_D'$, where $D$ is an involutory matrix satisfying the constraints of Theorem \ref{main theorem}.
	\begin{exl}\label{main example} Consider a  $2\times 5$ matrix $M=\begin{bmatrix}
			1 & w & w^2 & w^3 & w^4 \\
			w^3 &w^5 & w^7 & w^9 & w^{11}
		\end{bmatrix}
		$ over $\mathbb{F}_{16}$. Since the second row is not a scalar multiple of the first row, therefore, both rows are linearly independent and $M$ has rank $2.$ Now, we consider a $2\times 2$ involutory matrix $D=\begin{bmatrix}
			0 & w^5\\
			w^{10} & 0
		\end{bmatrix}$ over $\mathbb{F}_{16}$. It is easy to check that $D^{o4}=D$. Hence $D$ satisfies all the constraints of Theorem \ref{main theorem}. The concatenated matrix $G=[M\mid DM^{o4}J_6]$ is given by $$G=\begin{bmatrix}
			1 & w & w^2 & w^3 & w^4 & w^4 & w^{11} & w^3 & w^{10} & w^2\\
			w^3 & w^5 & w^7 & w^9 & w^{11} & w^{11} & w^7 & w^3 & w^{14} & w^{10}
		\end{bmatrix}.$$ Thus by Theorem \ref{main theorem} the $\Phi$ image of the linear code generated by the concatenated matrix $G=[M\mid DM^{o4}J_6]$ is a reversible DNA code. By SageMath computations \cite{Sage}, the code parameters of the linear code generated by $G$ are $[10,2,8]$. Hence $G$ generates an AMDS, that is, almost maximum distance separable linear code. We explicitly write the 256 linear and corresponding DNA codewords for the generator matrix $G.$
		
		\begin{longtable}{|c|l|l|}\caption{DNA correspondence table for Example \ref{main example}}\\
			\hline
			No. & Linear Codeword & DNA Codeword \\
			\hline
			\endfirsthead
			\hline
			No. & Linear Codeword & DNA Codeword \\
			\hline
			\endhead
			\hline
			1 & $0 0 0 0 0 0 0 0 0 0$ & \texttt{AAAAAAAAAAAAAAAAAAAA} \\
			\hline
			2 & $1 w^{1} w^{2} w^{3} w^{4} w^{4} w^{11} w^{3} w^{10} w^{2}$ & \texttt{TTATGCAGTATACTAGGGGC} \\
			\hline
			3 & $w^{1} w^{2} w^{3} w^{4} w^{5} w^{5} w^{12} w^{4} w^{11} w^{3}$ & \texttt{ATGCAGTACCCCGATACTAG} \\
			\hline
			4 & $w^{4} w^{5} w^{6} w^{7} w^{8} w^{8} 1 w^{7} w^{14} w^{6}$ & \texttt{TACCACGTCGCGTTGTTCAC} \\
			\hline
			5 & $w^{2} w^{3} w^{4} w^{5} w^{6} w^{6} w^{13} w^{5} w^{12} w^{4}$ & \texttt{GCAGTACCACACTGCCGATA} \\
			\hline
			6 & $w^{8} w^{9} w^{10} w^{11} w^{12} w^{12} w^{4} w^{11} w^{3} w^{10}$ & \texttt{CGCAGGCTGAGATACTAGGG} \\
			\hline
			7 & $w^{5} w^{6} w^{7} w^{8} w^{9} w^{9} w^{1} w^{8} 1 w^{7}$ & \texttt{CCACGTCGCACAATCGTTGT} \\
			\hline
			8 & $w^{10} w^{11} w^{12} w^{13} w^{14} w^{14} w^{6} w^{13} w^{5} w^{12}$ & \texttt{GGCTGATGTCTCACTGCCGA} \\
			\hline
			9 & $w^{3} w^{4} w^{5} w^{6} w^{7} w^{7} w^{14} w^{6} w^{13} w^{5}$ & \texttt{AGTACCACGTGTTCACTGCC} \\
			\hline
			10 & $w^{14} 1 w^{1} w^{2} w^{3} w^{3} w^{10} w^{2} w^{9} w^{1}$ & \texttt{TCTTATGCAGAGGGGCCAAT} \\
			\hline
			11 & $w^{9} w^{10} w^{11} w^{12} w^{13} w^{13} w^{5} w^{12} w^{4} w^{11}$ & \texttt{CAGGCTGATGTGCCGATACT} \\
			\hline
			12 & $w^{7} w^{8} w^{9} w^{10} w^{11} w^{11} w^{3} w^{10} w^{2} w^{9}$ & \texttt{GTCGCAGGCTCTAGGGGCCA} \\
			\hline
			13 & $w^{6} w^{7} w^{8} w^{9} w^{10} w^{10} w^{2} w^{9} w^{1} w^{8}$ & \texttt{ACGTCGCAGGGGGCCAATCG} \\
			\hline
			14 & $w^{13} w^{14} 1 w^{1} w^{2} w^{2} w^{9} w^{1} w^{8} 1$ & \texttt{TGTCTTATGCGCCAATCGTT} \\
			\hline
			15 & $w^{11} w^{12} w^{13} w^{14} 1 1 w^{7} w^{14} w^{6} w^{13}$ & \texttt{CTGATGTCTTTTGTTCACTG} \\
			\hline
			16 & $w^{12} w^{13} w^{14} 1 w^{1} w^{1} w^{8} 1 w^{7} w^{14}$ & \texttt{GATGTCTTATATCGTTGTTC} \\
			\hline
			17 & $w^{3} w^{5} w^{7} w^{9} w^{11} w^{11} w^{7} w^{3} w^{14} w^{10}$ & \texttt{AGCCGTCACTCTGTAGTCGG} \\
			\hline
			18 & $w^{14} w^{2} w^{12} w^{1} w^{13} w^{13} w^{8} 0 w^{11} w^{4}$ & \texttt{TCGCGAATTGTGCGAACTTA} \\
			\hline
			19 & $w^{9} w^{1} w^{4} w^{14} w^{3} w^{3} w^{2} w^{7} w^{10} w^{12}$ & \texttt{CAATTATCAGAGGCGTGGGA} \\
			\hline
			20 & $w^{7} 0 w^{10} 1 w^{7} w^{7} w^{9} w^{4} 0 w^{7}$ & \texttt{GTAAGGTTGTGTCATAAAGT} \\
			\hline
			21 & $w^{6} w^{11} w^{3} w^{6} w^{1} w^{1} w^{5} w^{11} w^{5} w^{2}$ & \texttt{ACCTAGACATATCCCTCCGC} \\
			\hline
			22 & $w^{13} w^{6} w^{6} w^{2} 1 1 w^{3} w^{5} 1 0$ & \texttt{TGACACGCTTTTAGCCTTAA} \\
			\hline
			23 & $w^{11} w^{9} 0 w^{12} w^{2} w^{2} w^{14} w^{13} w^{3} w^{6}$ & \texttt{CTCAAAGAGCGCTCTGAGAC} \\
			\hline
			24 & $w^{12} w^{3} w^{2} w^{10} w^{10} w^{10} w^{10} w^{8} w^{12} w^{3}$ & \texttt{GAAGGCGGGGGGGGCGGAAG} \\
			\hline
			25 & $0 w^{8} w^{13} w^{5} w^{8} w^{8} w^{1} w^{2} w^{2} 1$ & \texttt{AACGTGCCCGCGATGCGCTT} \\
			\hline
			26 & $1 w^{10} w^{14} w^{11} w^{5} w^{5} w^{6} w^{6} w^{4} w^{8}$ & \texttt{TTGGTCCTCCCCACACTACG} \\
			\hline
			27 & $w^{1} 1 w^{8} w^{8} w^{4} w^{4} w^{13} w^{10} w^{9} w^{14}$ & \texttt{ATTTCGCGTATATGGGCATC} \\
			\hline
			28 & $w^{4} w^{4} 1 w^{13} 0 0 w^{4} w^{12} w^{13} w^{13}$ & \texttt{TATATTTGAAAATAGATGTG} \\
			\hline
			29 & $w^{2} w^{13} w^{11} 0 w^{14} w^{14} w^{12} w^{1} w^{7} w^{1}$ & \texttt{GCTGCTAATCTCGAATGTAT} \\
			\hline
			30 & $w^{8} w^{12} w^{9} w^{3} w^{9} w^{9} 1 w^{9} w^{6} w^{5}$ & \texttt{CGGACAAGCACATTCAACCC} \\
			\hline
			31 & $w^{5} w^{14} w^{5} w^{4} w^{12} w^{12} 0 1 w^{8} w^{9}$ & \texttt{CCTCCCTAGAGAAATTCGCA} \\
			\hline
			32 & $w^{10} w^{7} w^{1} w^{7} w^{6} w^{6} w^{11} w^{14} w^{1} w^{11}$ & \texttt{GGGTATGTACACCTTCATCT} \\
			\hline
			33 & $w^{4} w^{6} w^{8} w^{10} w^{12} w^{12} w^{8} w^{4} 1 w^{11}$ & \texttt{TAACCGGGGAGACGTATTCT} \\
			\hline
			34 & $w^{1} w^{11} 1 w^{12} w^{6} w^{6} w^{7} w^{7} w^{5} w^{9}$ & \texttt{ATCTTTGAACACGTGTCCCA} \\
			\hline
			35 & $1 w^{3} w^{13} w^{2} w^{14} w^{14} w^{9} 0 w^{12} w^{5}$ & \texttt{TTAGTGGCTCTCCAAAGACC} \\
			\hline
			36 & $0 w^{9} w^{14} w^{6} w^{9} w^{9} w^{2} w^{3} w^{3} w^{1}$ & \texttt{AACATCACCACAGCAGAGAT} \\
			\hline
			37 & $w^{10} w^{2} w^{5} 1 w^{4} w^{4} w^{3} w^{8} w^{11} w^{13}$ & \texttt{GGGCCCTTTATAAGCGCTTG} \\
			\hline
			38 & $w^{5} w^{5} w^{1} w^{14} 0 0 w^{5} w^{13} w^{14} w^{14}$ & \texttt{CCCCATTCAAAACCTGTCTC} \\
			\hline
			39 & $w^{8} 0 w^{11} w^{1} w^{8} w^{8} w^{10} w^{5} 0 w^{8}$ & \texttt{CGAACTATCGCGGGCCAACG} \\
			\hline
			40 & $w^{2} w^{1} w^{9} w^{9} w^{5} w^{5} w^{14} w^{11} w^{10} 1$ & \texttt{GCATCACACCCCTCCTGGTT} \\
			\hline
			41 & $w^{7} w^{12} w^{4} w^{7} w^{2} w^{2} w^{6} w^{12} w^{6} w^{3}$ & \texttt{GTGATAGTGCGCACGAACAG} \\
			\hline
			42 & $w^{9} w^{13} w^{10} w^{4} w^{10} w^{10} w^{1} w^{10} w^{7} w^{6}$ & \texttt{CATGGGTAGGGGATGGGTAC} \\
			\hline
			43 & $w^{14} w^{7} w^{7} w^{3} w^{1} w^{1} w^{4} w^{6} w^{1} 0$ & \texttt{TCGTGTAGATATTAACATAA} \\
			\hline
			44 & $w^{3} w^{14} w^{12} 0 1 1 w^{13} w^{2} w^{8} w^{2}$ & \texttt{AGTCGAAATTTTTGGCCGGC} \\
			\hline
			45 & $w^{12} w^{10} 0 w^{13} w^{3} w^{3} 1 w^{14} w^{4} w^{7}$ & \texttt{GAGGAATGAGAGTTTCTAGT} \\
			\hline
			46 & $w^{11} w^{8} w^{2} w^{8} w^{7} w^{7} w^{12} 1 w^{2} w^{12}$ & \texttt{CTCGGCCGGTGTGATTGCGA} \\
			\hline
			47 & $w^{13} w^{4} w^{3} w^{11} w^{11} w^{11} w^{11} w^{9} w^{13} w^{4}$ & \texttt{TGTAAGCTCTCTCTCATGTA} \\
			\hline
			48 & $w^{6} 1 w^{6} w^{5} w^{13} w^{13} 0 w^{1} w^{9} w^{10}$ & \texttt{ACTTACCCTGTGAAATCAGG} \\
			\hline
			49 & $w^{7} w^{9} w^{11} w^{13} 1 1 w^{11} w^{7} w^{3} w^{14}$ & \texttt{GTCACTTGTTTTCTGTAGTC} \\
			\hline
			50 & $w^{9} w^{3} w^{9} w^{8} w^{1} w^{1} 0 w^{4} w^{12} w^{13}$ & \texttt{CAAGCACGATATAATAGATG} \\
			\hline
			51 & $w^{14} w^{11} w^{5} w^{11} w^{10} w^{10} 1 w^{3} w^{5} 1$ & \texttt{TCCTCCCTGGGGTTAGCCTT} \\
			\hline
			52 & $w^{3} w^{6} w^{1} w^{5} w^{2} w^{2} w^{12} 0 1 w^{8}$ & \texttt{AGACATCCGCGCGAAATTCG} \\
			\hline
			53 & $w^{12} w^{1} w^{13} w^{7} w^{13} w^{13} w^{4} w^{13} w^{10} w^{9}$ & \texttt{GAATTGGTTGTGTATGGGCA} \\
			\hline
			54 & $w^{11} 0 w^{14} w^{4} w^{11} w^{11} w^{13} w^{8} 0 w^{11}$ & \texttt{CTAATCTACTCTTGCGAACT} \\
			\hline
			55 & $w^{13} w^{5} w^{8} w^{3} w^{7} w^{7} w^{6} w^{11} w^{14} w^{1}$ & \texttt{TGCCCGAGGTGTACCTTCAT} \\
			\hline
			56 & $w^{6} w^{2} 1 0 w^{3} w^{3} w^{1} w^{5} w^{11} w^{5}$ & \texttt{ACGCTTAAAGAGATCCCTCC} \\
			\hline
			57 & $w^{4} w^{14} w^{3} 1 w^{9} w^{9} w^{10} w^{10} w^{8} w^{12}$ & \texttt{TATCAGTTCACAGGGGCGGA} \\
			\hline
			58 & $w^{1} w^{7} w^{6} w^{14} w^{14} w^{14} w^{14} w^{12} w^{1} w^{7}$ & \texttt{ATGTACTCTCTCTCGAATGT} \\
			\hline
			59 & $1 w^{13} 0 w^{1} w^{6} w^{6} w^{3} w^{2} w^{7} w^{10}$ & \texttt{TTTGAAATACACAGGCGTGG} \\
			\hline
			60 & $0 w^{12} w^{2} w^{9} w^{12} w^{12} w^{5} w^{6} w^{6} w^{4}$ & \texttt{AAGAGCCAGAGACCACACTA} \\
			\hline
			61 & $w^{10} 1 w^{7} w^{10} w^{5} w^{5} w^{9} 1 w^{9} w^{6}$ & \texttt{GGTTGTGGCCCCCATTCAAC} \\
			\hline
			62 & $w^{5} w^{4} w^{12} w^{12} w^{8} w^{8} w^{2} w^{14} w^{13} w^{3}$ & \texttt{CCTAGAGACGCGGCTCTGAG} \\
			\hline
			63 & $w^{8} w^{8} w^{4} w^{2} 0 0 w^{8} w^{1} w^{2} w^{2}$ & \texttt{CGCGTAGCAAAACGATGCGC} \\
			\hline
			64 & $w^{2} w^{10} w^{10} w^{6} w^{4} w^{4} w^{7} w^{9} w^{4} 0$ & \texttt{GCGGGGACTATAGTCATAAA} \\
			\hline
			65 & $w^{5} w^{7} w^{9} w^{11} w^{13} w^{13} w^{9} w^{5} w^{1} w^{12}$ & \texttt{CCGTCACTTGTGCACCATGA} \\
			\hline
			66 & $w^{10} w^{14} w^{11} w^{5} w^{11} w^{11} w^{2} w^{11} w^{8} w^{7}$ & \texttt{GGTCCTCCCTCTGCCTCGGT} \\
			\hline
			67 & $w^{2} w^{12} w^{1} w^{13} w^{7} w^{7} w^{8} w^{8} w^{6} w^{10}$ & \texttt{GCGAATTGGTGTCGCGACGG} \\
			\hline
			68 & $w^{8} w^{13} w^{5} w^{8} w^{3} w^{3} w^{7} w^{13} w^{7} w^{4}$ & \texttt{CGTGCCCGAGAGGTTGGTTA} \\
			\hline
			69 & $w^{1} w^{4} w^{14} w^{3} 1 1 w^{10} 0 w^{13} w^{6}$ & \texttt{ATTATCAGTTTTGGAATGAC} \\
			\hline
			70 & $w^{4} 1 w^{13} 0 w^{1} w^{1} w^{14} w^{3} w^{9} w^{3}$ & \texttt{TATTTGAAATATTCAGCAAG} \\
			\hline
			71 & $0 w^{10} 1 w^{7} w^{10} w^{10} w^{3} w^{4} w^{4} w^{2}$ & \texttt{AAGGTTGTGGGGAGTATAGC} \\
			\hline
			72 & $1 w^{8} w^{8} w^{4} w^{2} w^{2} w^{5} w^{7} w^{2} 0$ & \texttt{TTCGCGTAGCGCCCGTGCAA} \\
			\hline
			73 & $w^{11} w^{3} w^{6} w^{1} w^{5} w^{5} w^{4} w^{9} w^{12} w^{14}$ & \texttt{CTAGACATCCCCTACAGATC} \\
			\hline
			74 & $w^{12} w^{9} w^{3} w^{9} w^{8} w^{8} w^{13} w^{1} w^{3} w^{13}$ & \texttt{GACAAGCACGCGTGATAGTG} \\
			\hline
			75 & $w^{6} w^{6} w^{2} 1 0 0 w^{6} w^{14} 1 1$ & \texttt{ACACGCTTAAAAACTCTTTT} \\
			\hline
			76 & $w^{13} w^{11} 0 w^{14} w^{4} w^{4} w^{1} 1 w^{5} w^{8}$ & \texttt{TGCTAATCTATAATTTCCCG} \\
			\hline
			77 & $w^{9} 0 w^{12} w^{2} w^{9} w^{9} w^{11} w^{6} 0 w^{9}$ & \texttt{CAAAGAGCCACACTACAACA} \\
			\hline
			78 & $w^{7} w^{1} w^{7} w^{6} w^{14} w^{14} 0 w^{2} w^{10} w^{11}$ & \texttt{GTATGTACTCTCAAGCGGCT} \\
			\hline
			79 & $w^{3} w^{2} w^{10} w^{10} w^{6} w^{6} 1 w^{12} w^{11} w^{1}$ & \texttt{AGGCGGGGACACTTGACTAT} \\
			\hline
			80 & $w^{14} w^{5} w^{4} w^{12} w^{12} w^{12} w^{12} w^{10} w^{14} w^{5}$ & \texttt{TCCCTAGAGAGAGAGGTCCC} \\
			\hline
			81 & $w^{11} w^{13} 1 w^{2} w^{4} w^{4} 1 w^{11} w^{7} w^{3}$ & \texttt{CTTGTTGCTATATTCTGTAG} \\
			\hline
			82 & $w^{12} w^{12} w^{8} w^{6} 0 0 w^{12} w^{5} w^{6} w^{6}$ & \texttt{GAGACGACAAAAGACCACAC} \\
			\hline
			83 & $w^{6} w^{14} w^{14} w^{10} w^{8} w^{8} w^{11} w^{13} w^{8} 0$ & \texttt{ACTCTCGGCGCGCTTGCGAA} \\
			\hline
			84 & $w^{13} w^{7} w^{13} w^{12} w^{5} w^{5} 0 w^{8} w^{1} w^{2}$ & \texttt{TGGTTGGACCCCAACGATGC} \\
			\hline
			85 & $w^{9} w^{8} w^{1} w^{1} w^{12} w^{12} w^{6} w^{3} w^{2} w^{7}$ & \texttt{CACGATATGAGAACAGGCGT} \\
			\hline
			86 & $w^{7} w^{10} w^{5} w^{9} w^{6} w^{6} w^{1} 0 w^{4} w^{12}$ & \texttt{GTGGCCCAACACATAATAGA} \\
			\hline
			87 & $w^{3} 1 w^{9} 1 w^{14} w^{14} w^{4} w^{7} w^{9} w^{4}$ & \texttt{AGTTCATTTCTCTAGTCATA} \\
			\hline
			88 & $w^{14} w^{4} w^{11} w^{14} w^{9} w^{9} w^{13} w^{4} w^{13} w^{10}$ & \texttt{TCTACTTCCACATGTATGGG} \\
			\hline
			89 & $w^{5} w^{11} w^{10} w^{3} w^{3} w^{3} w^{3} w^{1} w^{5} w^{11}$ & \texttt{CCCTGGAGAGAGAGATCCCT} \\
			\hline
			90 & $w^{10} w^{6} w^{4} 0 w^{7} w^{7} w^{5} w^{9} 1 w^{9}$ & \texttt{GGACTAAAGTGTCCCATTCA} \\
			\hline
			91 & $w^{2} w^{9} w^{12} w^{7} w^{11} w^{11} w^{10} 1 w^{3} w^{5}$ & \texttt{GCCAGAGTCTCTGGTTAGCC} \\
			\hline
			92 & $w^{8} w^{3} w^{7} w^{4} w^{13} w^{13} w^{14} w^{14} w^{12} w^{1}$ & \texttt{CGAGGTTATGTGTCTCGAAT} \\
			\hline
			93 & $w^{1} w^{5} w^{2} w^{11} w^{2} w^{2} w^{8} w^{2} w^{14} w^{13}$ & \texttt{ATCCGCCTGCGCCGGCTCTG} \\
			\hline
			94 & $w^{4} w^{2} 0 w^{5} w^{10} w^{10} w^{7} w^{6} w^{11} w^{14}$ & \texttt{TAGCAACCGGGGGTACCTTC} \\
			\hline
			95 & $0 w^{1} w^{6} w^{13} w^{1} w^{1} w^{9} w^{10} w^{10} w^{8}$ & \texttt{AAATACTGATATCAGGGGCG} \\
			\hline
			96 & $1 0 w^{3} w^{8} 1 1 w^{2} w^{12} 0 1$ & \texttt{TTAAAGCGTTTTGCGAAATT} \\
			\hline
			97 & $w^{8} w^{10} w^{12} w^{14} w^{1} w^{1} w^{12} w^{8} w^{4} 1$ & \texttt{CGGGGATCATATGACGTATT} \\
			\hline
			98 & $w^{2} w^{8} w^{7} 1 1 1 1 w^{13} w^{2} w^{8}$ & \texttt{GCCGGTTTTTTTTTTGGCCG} \\
			\hline
			99 & $w^{10} w^{4} w^{10} w^{9} w^{2} w^{2} 0 w^{5} w^{13} w^{14}$ & \texttt{GGTAGGCAGCGCAACCTGTC} \\
			\hline
			100 & $w^{5} 1 w^{4} w^{1} w^{10} w^{10} w^{11} w^{11} w^{9} w^{13}$ & \texttt{CCTTTAATGGGGCTCTCATG} \\
			\hline
			101 & $1 w^{12} w^{6} w^{12} w^{11} w^{11} w^{1} w^{4} w^{6} w^{1}$ & \texttt{TTGAACGACTCTATTAACAT} \\
			\hline
			102 & $0 w^{13} w^{3} w^{10} w^{13} w^{13} w^{6} w^{7} w^{7} w^{5}$ & \texttt{AATGAGGGTGTGACGTGTCC} \\
			\hline
			103 & $w^{4} w^{7} w^{2} w^{6} w^{3} w^{3} w^{13} 0 w^{1} w^{9}$ & \texttt{TAGTGCACAGAGTGAAATCA} \\
			\hline
			104 & $w^{1} w^{14} 0 w^{2} w^{7} w^{7} w^{4} w^{3} w^{8} w^{11}$ & \texttt{ATTCAAGCGTGTTAAGCGCT} \\
			\hline
			105 & $w^{13} w^{2} w^{14} w^{8} w^{14} w^{14} w^{5} w^{14} w^{11} w^{10}$ & \texttt{TGGCTCCGTCTCCCTCCTGG} \\
			\hline
			106 & $w^{6} w^{5} w^{13} w^{13} w^{9} w^{9} w^{3} 1 w^{14} w^{4}$ & \texttt{ACCCTGTGCACAAGTTTCTA} \\
			\hline
			107 & $w^{12} 0 1 w^{5} w^{12} w^{12} w^{14} w^{9} 0 w^{12}$ & \texttt{GAAATTCCGAGATCCAAAGA} \\
			\hline
			108 & $w^{11} w^{1} w^{8} w^{11} w^{6} w^{6} w^{10} w^{1} w^{10} w^{7}$ & \texttt{CTATCGCTACACGGATGGGT} \\
			\hline
			109 & $w^{14} w^{6} w^{9} w^{4} w^{8} w^{8} w^{7} w^{12} 1 w^{2}$ & \texttt{TCACCATACGCGGTGATTGC} \\
			\hline
			110 & $w^{3} w^{11} w^{11} w^{7} w^{5} w^{5} w^{8} w^{10} w^{5} 0$ & \texttt{AGCTCTGTCCCCCGGGCCAA} \\
			\hline
			111 & $w^{7} w^{3} w^{1} 0 w^{4} w^{4} w^{2} w^{6} w^{12} w^{6}$ & \texttt{GTAGATAATATAGCACGAAC} \\
			\hline
			112 & $w^{9} w^{9} w^{5} w^{3} 0 0 w^{9} w^{2} w^{3} w^{3}$ & \texttt{CACACCAGAAAACAGCAGAG} \\
			\hline
			113 & $w^{13} 1 w^{2} w^{4} w^{6} w^{6} w^{2} w^{13} w^{9} w^{5}$ & \texttt{TGTTGCTAACACGCTGCACC} \\
			\hline
			114 & $w^{6} w^{4} 0 w^{7} w^{12} w^{12} w^{9} w^{8} w^{13} w^{1}$ & \texttt{ACTAAAGTGAGACACGTGAT} \\
			\hline
			115 & $w^{12} w^{8} w^{6} 0 w^{9} w^{9} w^{7} w^{11} w^{2} w^{11}$ & \texttt{GACGACAACACAGTCTGCCT} \\
			\hline
			116 & $w^{11} w^{10} w^{3} w^{3} w^{14} w^{14} w^{8} w^{5} w^{4} w^{9}$ & \texttt{CTGGAGAGTCTCCGCCTACA} \\
			\hline
			117 & $w^{14} w^{14} w^{10} w^{8} 0 0 w^{14} w^{7} w^{8} w^{8}$ & \texttt{TCTCGGCGAAAATCGTCGCG} \\
			\hline
			118 & $w^{3} w^{7} w^{4} w^{13} w^{4} w^{4} w^{10} w^{4} w^{1} 1$ & \texttt{AGGTTATGTATAGGTAATTT} \\
			\hline
			119 & $w^{7} w^{13} w^{12} w^{5} w^{5} w^{5} w^{5} w^{3} w^{7} w^{13}$ & \texttt{GTTGGACCCCCCCCAGGTTG} \\
			\hline
			120 & $w^{9} w^{12} w^{7} w^{11} w^{8} w^{8} w^{3} 0 w^{6} w^{14}$ & \texttt{CAGAGTCTCGCGAGAAACTC} \\
			\hline
			121 & $w^{8} w^{1} w^{1} w^{12} w^{10} w^{10} w^{13} 1 w^{10} 0$ & \texttt{CGATATGAGGGGTGTTGGAA} \\
			\hline
			122 & $w^{2} 0 w^{5} w^{10} w^{2} w^{2} w^{4} w^{14} 0 w^{2}$ & \texttt{GCAACCGGGCGCTATCAAGC} \\
			\hline
			123 & $w^{10} w^{5} w^{9} w^{6} 1 1 w^{1} w^{1} w^{14} w^{3}$ & \texttt{GGCCCAACTTTTATATTCAG} \\
			\hline
			124 & $w^{5} w^{2} w^{11} w^{2} w^{1} w^{1} w^{6} w^{9} w^{11} w^{6}$ & \texttt{CCGCCTGCATATACCACTAC} \\
			\hline
			125 & $1 w^{9} 1 w^{14} w^{7} w^{7} 0 w^{10} w^{3} w^{4}$ & \texttt{TTCATTTCGTGTAAGGAGTA} \\
			\hline
			126 & $0 w^{3} w^{8} 1 w^{3} w^{3} w^{11} w^{12} w^{12} w^{10}$ & \texttt{AAAGCGTTAGAGCTGAGAGG} \\
			\hline
			127 & $w^{4} w^{11} w^{14} w^{9} w^{13} w^{13} w^{12} w^{2} w^{5} w^{7}$ & \texttt{TACTTCCATGTGGAGCCCGT} \\
			\hline
			128 & $w^{1} w^{6} w^{13} w^{1} w^{11} w^{11} 1 w^{6} 1 w^{12}$ & \texttt{ATACTGATCTCTTTACTTGA} \\
			\hline
			129 & $w^{6} w^{8} w^{10} w^{12} w^{14} w^{14} w^{10} w^{6} w^{2} w^{13}$ & \texttt{ACCGGGGATCTCGGACGCTG} \\
			\hline
			130 & $w^{13} w^{10} w^{4} w^{10} w^{9} w^{9} w^{14} w^{2} w^{4} w^{14}$ & \texttt{TGGGTAGGCACATCGCTATC} \\
			\hline
			131 & $w^{11} 1 w^{12} w^{6} w^{12} w^{12} w^{3} w^{12} w^{9} w^{8}$ & \texttt{CTTTGAACGAGAAGGACACG} \\
			\hline
			132 & $w^{12} w^{4} w^{7} w^{2} w^{6} w^{6} w^{5} w^{10} w^{13} 1$ & \texttt{GATAGTGCACACCCGGTGTT} \\
			\hline
			133 & $w^{3} w^{13} w^{2} w^{14} w^{8} w^{8} w^{9} w^{9} w^{7} w^{11}$ & \texttt{AGTGGCTCCGCGCACAGTCT} \\
			\hline
			134 & $w^{14} w^{12} 0 1 w^{5} w^{5} w^{2} w^{1} w^{6} w^{9}$ & \texttt{TCGAAATTCCCCGCATACCA} \\
			\hline
			135 & $w^{9} w^{14} w^{6} w^{9} w^{4} w^{4} w^{8} w^{14} w^{8} w^{5}$ & \texttt{CATCACCATATACGTCCGCC} \\
			\hline
			136 & $w^{7} w^{7} w^{3} w^{1} 0 0 w^{7} 1 w^{1} w^{1}$ & \texttt{GTGTAGATAAAAGTTTATAT} \\
			\hline
			137 & $w^{2} w^{5} 1 w^{4} w^{1} w^{1} w^{11} 0 w^{14} w^{7}$ & \texttt{GCCCTTTAATATCTAATCGT} \\
			\hline
			138 & $w^{8} w^{2} w^{8} w^{7} 1 1 0 w^{3} w^{11} w^{12}$ & \texttt{CGGCCGGTTTTTAAAGCTGA} \\
			\hline
			139 & $w^{5} w^{1} w^{14} 0 w^{2} w^{2} 1 w^{4} w^{10} w^{4}$ & \texttt{CCATTCAAGCGCTTTAGGTA} \\
			\hline
			140 & $w^{10} 0 w^{13} w^{3} w^{10} w^{10} w^{12} w^{7} 0 w^{10}$ & \texttt{GGAATGAGGGGGGAGTAAGG} \\
			\hline
			141 & $0 w^{11} w^{1} w^{8} w^{11} w^{11} w^{4} w^{5} w^{5} w^{3}$ & \texttt{AACTATCGCTCTTACCCCAG} \\
			\hline
			142 & $1 w^{6} w^{5} w^{13} w^{13} w^{13} w^{13} w^{11} 1 w^{6}$ & \texttt{TTACCCTGTGTGTGCTTTAC} \\
			\hline
			143 & $w^{1} w^{9} w^{9} w^{5} w^{3} w^{3} w^{6} w^{8} w^{3} 0$ & \texttt{ATCACACCAGAGACCGAGAA} \\
			\hline
			144 & $w^{4} w^{3} w^{11} w^{11} w^{7} w^{7} w^{1} w^{13} w^{12} w^{2}$ & \texttt{TAAGCTCTGTGTATTGGAGC} \\
			\hline
			145 & $w^{2} w^{4} w^{6} w^{8} w^{10} w^{10} w^{6} w^{2} w^{13} w^{9}$ & \texttt{GCTAACCGGGGGACGCTGCA} \\
			\hline
			146 & $w^{8} 1 w^{3} w^{13} w^{2} w^{2} w^{1} w^{6} w^{9} w^{11}$ & \texttt{CGTTAGTGGCGCATACCACT} \\
			\hline
			147 & $w^{5} w^{10} w^{2} w^{5} 1 1 w^{4} w^{10} w^{4} w^{1}$ & \texttt{CCGGGCCCTTTTTAGGTAAT} \\
			\hline
			148 & $w^{10} w^{8} 0 w^{11} w^{1} w^{1} w^{13} w^{12} w^{2} w^{5}$ & \texttt{GGCGAACTATATTGGAGCCC} \\
			\hline
			149 & $0 w^{7} w^{12} w^{4} w^{7} w^{7} 1 w^{1} w^{1} w^{14}$ & \texttt{AAGTGATAGTGTTTATATTC} \\
			\hline
			150 & $1 w^{14} w^{7} w^{7} w^{3} w^{3} w^{12} w^{9} w^{8} w^{13}$ & \texttt{TTTCGTGTAGAGGACACGTG} \\
			\hline
			151 & $w^{1} w^{12} w^{10} 0 w^{13} w^{13} w^{11} 1 w^{6} 1$ & \texttt{ATGAGGAATGTGCTTTACTT} \\
			\hline
			152 & $w^{4} w^{13} w^{4} w^{3} w^{11} w^{11} 0 w^{14} w^{7} w^{8}$ & \texttt{TATGTAAGCTCTAATCGTCG} \\
			\hline
			153 & $w^{6} 0 w^{9} w^{14} w^{6} w^{6} w^{8} w^{3} 0 w^{6}$ & \texttt{ACAACATCACACCGAGAAAC} \\
			\hline
			154 & $w^{13} w^{1} w^{11} 1 w^{12} w^{12} w^{7} 0 w^{10} w^{3}$ & \texttt{TGATCTTTGAGAGTAAGGAG} \\
			\hline
			155 & $w^{11} w^{2} w^{1} w^{9} w^{9} w^{9} w^{9} w^{7} w^{11} w^{2}$ & \texttt{CTGCATCACACACAGTCTGC} \\
			\hline
			156 & $w^{12} w^{5} w^{5} w^{1} w^{14} w^{14} w^{2} w^{4} w^{14} 0$ & \texttt{GACCCCATTCTCGCTATCAA} \\
			\hline
			157 & $w^{3} w^{3} w^{14} w^{12} 0 0 w^{3} w^{11} w^{12} w^{12}$ & \texttt{AGAGTCGAAAAAAGCTGAGA} \\
			\hline
			158 & $w^{14} w^{9} w^{13} w^{10} w^{4} w^{4} w^{5} w^{5} w^{3} w^{7}$ & \texttt{TCCATGGGTATACCCCAGGT} \\
			\hline
			159 & $w^{9} w^{6} 1 w^{6} w^{5} w^{5} w^{10} w^{13} 1 w^{10}$ & \texttt{CAACTTACCCCCGGTGTTGG} \\
			\hline
			160 & $w^{7} w^{11} w^{8} w^{2} w^{8} w^{8} w^{14} w^{8} w^{5} w^{4}$ & \texttt{GTCTCGGCCGCGTCCGCCTA} \\
			\hline
			161 & $w^{12} w^{14} w^{1} w^{3} w^{5} w^{5} w^{1} w^{12} w^{8} w^{4}$ & \texttt{GATCATAGCCCCATGACGTA} \\
			\hline
			162 & $w^{11} w^{7} w^{5} 0 w^{8} w^{8} w^{6} w^{10} w^{1} w^{10}$ & \texttt{CTGTCCAACGCGACGGATGG} \\
			\hline
			163 & $w^{13} w^{13} w^{9} w^{7} 0 0 w^{13} w^{6} w^{7} w^{7}$ & \texttt{TGTGCAGTAAAATGACGTGT} \\
			\hline
			164 & $w^{6} w^{12} w^{11} w^{4} w^{4} w^{4} w^{4} w^{2} w^{6} w^{12}$ & \texttt{ACGACTTATATATAGCACGA} \\
			\hline
			165 & $w^{7} 1 1 w^{11} w^{9} w^{9} w^{12} w^{14} w^{9} 0$ & \texttt{GTTTTTCTCACAGATCCAAA} \\
			\hline
			166 & $w^{9} w^{4} w^{8} w^{5} w^{14} w^{14} 1 1 w^{13} w^{2}$ & \texttt{CATACGCCTCTCTTTTTGGC} \\
			\hline
			167 & $w^{14} w^{8} w^{14} w^{13} w^{6} w^{6} 0 w^{9} w^{2} w^{3}$ & \texttt{TCCGTCTGACACAACAGCAG} \\
			\hline
			168 & $w^{3} w^{10} w^{13} w^{8} w^{12} w^{12} w^{11} w^{1} w^{4} w^{6}$ & \texttt{AGGGTGCGGAGACTATTAAC} \\
			\hline
			169 & $w^{10} w^{9} w^{2} w^{2} w^{13} w^{13} w^{7} w^{4} w^{3} w^{8}$ & \texttt{GGCAGCGCTGTGGTTAAGCG} \\
			\hline
			170 & $w^{5} w^{3} 0 w^{6} w^{11} w^{11} w^{8} w^{7} w^{12} 1$ & \texttt{CCAGAAACCTCTCGGTGATT} \\
			\hline
			171 & $w^{8} w^{11} w^{6} w^{10} w^{7} w^{7} w^{2} 0 w^{5} w^{13}$ & \texttt{CGCTACGGGTGTGCAACCTG} \\
			\hline
			172 & $w^{2} w^{6} w^{3} w^{12} w^{3} w^{3} w^{9} w^{3} 1 w^{14}$ & \texttt{GCACAGGAAGAGCAAGTTTC} \\
			\hline
			173 & $w^{4} w^{1} w^{10} w^{1} 1 1 w^{5} w^{8} w^{10} w^{5}$ & \texttt{TAATGGATTTTTCCCGGGCC} \\
			\hline
			174 & $w^{1} 0 w^{4} w^{9} w^{1} w^{1} w^{3} w^{13} 0 w^{1}$ & \texttt{ATAATACAATATAGTGAAAT} \\
			\hline
			175 & $1 w^{5} w^{12} 1 w^{10} w^{10} w^{14} w^{5} w^{14} w^{11}$ & \texttt{TTCCGATTGGGGTCCCTCCT} \\
			\hline
			176 & $0 w^{2} w^{7} w^{14} w^{2} w^{2} w^{10} w^{11} w^{11} w^{9}$ & \texttt{AAGCGTTCGCGCGGCTCTCA} \\
			\hline
			177 & $w^{10} w^{12} w^{14} w^{1} w^{3} w^{3} w^{14} w^{10} w^{6} w^{2}$ & \texttt{GGGATCATAGAGTCGGACGC} \\
			\hline
			178 & $w^{5} w^{13} w^{13} w^{9} w^{7} w^{7} w^{10} w^{12} w^{7} 0$ & \texttt{CCTGTGCAGTGTGGGAGTAA} \\
			\hline
			179 & $w^{8} w^{7} 1 1 w^{11} w^{11} w^{5} w^{2} w^{1} w^{6}$ & \texttt{CGGTTTTTCTCTCCGCATAC} \\
			\hline
			180 & $w^{2} w^{14} w^{8} w^{14} w^{13} w^{13} w^{3} w^{6} w^{8} w^{3}$ & \texttt{GCTCCGTCTGTGAGACCGAG} \\
			\hline
			181 & $w^{4} w^{10} w^{9} w^{2} w^{2} w^{2} w^{2} 1 w^{4} w^{10}$ & \texttt{TAGGCAGCGCGCGCTTTAGG} \\
			\hline
			182 & $w^{1} w^{8} w^{11} w^{6} w^{10} w^{10} w^{9} w^{14} w^{2} w^{4}$ & \texttt{ATCGCTACGGGGCATCGCTA} \\
			\hline
			183 & $1 w^{4} w^{1} w^{10} w^{1} w^{1} w^{7} w^{1} w^{13} w^{12}$ & \texttt{TTTAATGGATATGTATTGGA} \\
			\hline
			184 & $0 1 w^{5} w^{12} 1 1 w^{8} w^{9} w^{9} w^{7}$ & \texttt{AATTCCGATTTTCGCACAGT} \\
			\hline
			185 & $w^{12} w^{6} w^{12} w^{11} w^{4} w^{4} 0 w^{7} 1 w^{1}$ & \texttt{GAACGACTTATAAAGTTTAT} \\
			\hline
			186 & $w^{11} w^{11} w^{7} w^{5} 0 0 w^{11} w^{4} w^{5} w^{5}$ & \texttt{CTCTGTCCAAAACTTACCCC} \\
			\hline
			187 & $w^{13} w^{3} w^{10} w^{13} w^{8} w^{8} w^{12} w^{3} w^{12} w^{9}$ & \texttt{TGAGGGTGCGCGGAAGGACA} \\
			\hline
			188 & $w^{6} w^{9} w^{4} w^{8} w^{5} w^{5} 1 0 w^{3} w^{11}$ & \texttt{ACCATACGCCCCTTAAAGCT} \\
			\hline
			189 & $w^{7} w^{2} w^{6} w^{3} w^{12} w^{12} w^{13} w^{13} w^{11} 1$ & \texttt{GTGCACAGGAGATGTGCTTT} \\
			\hline
			190 & $w^{9} w^{5} w^{3} 0 w^{6} w^{6} w^{4} w^{8} w^{14} w^{8}$ & \texttt{CACCAGAAACACTACGTCCG} \\
			\hline
			191 & $w^{14} 0 w^{2} w^{7} w^{14} w^{14} w^{1} w^{11} 0 w^{14}$ & \texttt{TCAAGCGTTCTCATCTAATC} \\
			\hline
			192 & $w^{3} w^{1} 0 w^{4} w^{9} w^{9} w^{6} w^{5} w^{10} w^{13}$ & \texttt{AGATAATACACAACCCGGTG} \\
			\hline
			193 & $w^{9} w^{11} w^{13} 1 w^{2} w^{2} w^{13} w^{9} w^{5} w^{1}$ & \texttt{CACTTGTTGCGCTGCACCAT} \\
			\hline
			194 & $w^{7} w^{6} w^{14} w^{14} w^{10} w^{10} w^{4} w^{1} 1 w^{5}$ & \texttt{GTACTCTCGGGGTAATTTCC} \\
			\hline
			195 & $w^{3} w^{9} w^{8} w^{1} w^{1} w^{1} w^{1} w^{14} w^{3} w^{9}$ & \texttt{AGCACGATATATATTCAGCA} \\
			\hline
			196 & $w^{14} w^{3} 1 w^{9} 1 1 w^{6} 1 w^{12} w^{11}$ & \texttt{TCAGTTCATTTTACTTGACT} \\
			\hline
			197 & $w^{11} w^{5} w^{11} w^{10} w^{3} w^{3} 0 w^{6} w^{14} 1$ & \texttt{CTCCCTGGAGAGAAACTCTT} \\
			\hline
			198 & $w^{12} w^{2} w^{9} w^{12} w^{7} w^{7} w^{11} w^{2} w^{11} w^{8}$ & \texttt{GAGCCAGAGTGTCTGCCTCG} \\
			\hline
			199 & $w^{6} w^{1} w^{5} w^{2} w^{11} w^{11} w^{12} w^{12} w^{10} w^{14}$ & \texttt{ACATCCGCCTCTGAGAGGTC} \\
			\hline
			200 & $w^{13} 0 w^{1} w^{6} w^{13} w^{13} 1 w^{10} 0 w^{13}$ & \texttt{TGAAATACTGTGTTGGAATG} \\
			\hline
			201 & $w^{1} w^{13} w^{7} w^{13} w^{12} w^{12} w^{2} w^{5} w^{7} w^{2}$ & \texttt{ATTGGTTGGAGAGCCCGTGC} \\
			\hline
			202 & $w^{4} w^{12} w^{12} w^{8} w^{6} w^{6} w^{9} w^{11} w^{6} 0$ & \texttt{TAGAGACGACACCACTACAA} \\
			\hline
			203 & $0 w^{14} w^{4} w^{11} w^{14} w^{14} w^{7} w^{8} w^{8} w^{6}$ & \texttt{AATCTACTTCTCGTCGCGAC} \\
			\hline
			204 & $1 w^{7} w^{10} w^{5} w^{9} w^{9} w^{8} w^{13} w^{1} w^{3}$ & \texttt{TTGTGGCCCACACGTGATAG} \\
			\hline
			205 & $w^{5} w^{8} w^{3} w^{7} w^{4} w^{4} w^{14} 0 w^{2} w^{10}$ & \texttt{CCCGAGGTTATATCAAGCGG} \\
			\hline
			206 & $w^{10} w^{10} w^{6} w^{4} 0 0 w^{10} w^{3} w^{4} w^{4}$ & \texttt{GGGGACTAAAAAGGAGTATA} \\
			\hline
			207 & $w^{2} 1 0 w^{3} w^{8} w^{8} w^{5} w^{4} w^{9} w^{12}$ & \texttt{GCTTAAAGCGCGCCTACAGA} \\
			\hline
			208 & $w^{8} w^{4} w^{2} 0 w^{5} w^{5} w^{3} w^{7} w^{13} w^{7}$ & \texttt{CGTAGCAACCCCAGGTTGGT} \\
			\hline
			209 & $w^{1} w^{3} w^{5} w^{7} w^{9} w^{9} w^{5} w^{1} w^{12} w^{8}$ & \texttt{ATAGCCGTCACACCATGACG} \\
			\hline
			210 & $w^{4} w^{9} w^{1} w^{4} w^{14} w^{14} w^{3} w^{9} w^{3} 1$ & \texttt{TACAATTATCTCAGCAAGTT} \\
			\hline
			211 & $0 w^{6} w^{11} w^{3} w^{6} w^{6} w^{14} 1 1 w^{13}$ & \texttt{AAACCTAGACACTCTTTTTG} \\
			\hline
			212 & $1 w^{11} w^{9} 0 w^{12} w^{12} w^{10} w^{14} w^{5} w^{14}$ & \texttt{TTCTCAAAGAGAGGTCCCTC} \\
			\hline
			213 & $w^{5} 0 w^{8} w^{13} w^{5} w^{5} w^{7} w^{2} 0 w^{5}$ & \texttt{CCAACGTGCCCCGTGCAACC} \\
			\hline
			214 & $w^{10} w^{1} 1 w^{8} w^{8} w^{8} w^{8} w^{6} w^{10} w^{1}$ & \texttt{GGATTTCGCGCGCGACGGAT} \\
			\hline
			215 & $w^{2} w^{2} w^{13} w^{11} 0 0 w^{2} w^{10} w^{11} w^{11}$ & \texttt{GCGCTGCTAAAAGCGGCTCT} \\
			\hline
			216 & $w^{8} w^{5} w^{14} w^{5} w^{4} w^{4} w^{9} w^{12} w^{14} w^{9}$ & \texttt{CGCCTCCCTATACAGATCCA} \\
			\hline
			217 & $w^{9} w^{7} 0 w^{10} 1 1 w^{12} w^{11} w^{1} w^{4}$ & \texttt{CAGTAAGGTTTTGACTATTA} \\
			\hline
			218 & $w^{7} w^{14} w^{2} w^{12} w^{1} w^{1} 1 w^{5} w^{8} w^{10}$ & \texttt{GTTCGCGAATATTTCCCGGG} \\
			\hline
			219 & $w^{3} w^{12} w^{3} w^{2} w^{10} w^{10} 0 w^{13} w^{6} w^{7}$ & \texttt{AGGAAGGCGGGGAATGACGT} \\
			\hline
			220 & $w^{14} w^{13} w^{6} w^{6} w^{2} w^{2} w^{11} w^{8} w^{7} w^{12}$ & \texttt{TCTGACACGCGCCTCGGTGA} \\
			\hline
			221 & $w^{11} w^{4} w^{4} 1 w^{13} w^{13} w^{1} w^{3} w^{13} 0$ & \texttt{CTTATATTTGTGATAGTGAA} \\
			\hline
			222 & $w^{12} 1 w^{10} w^{14} w^{11} w^{11} w^{6} 0 w^{9} w^{2}$ & \texttt{GATTGGTCCTCTACAACAGC} \\
			\hline
			223 & $w^{6} w^{10} w^{7} w^{1} w^{7} w^{7} w^{13} w^{7} w^{4} w^{3}$ & \texttt{ACGGGTATGTGTTGGTTAAG} \\
			\hline
			224 & $w^{13} w^{8} w^{12} w^{9} w^{3} w^{3} w^{4} w^{4} w^{2} w^{6}$ & \texttt{TGCGGACAAGAGTATAGCAC} \\
			\hline
			225 & $w^{14} w^{1} w^{3} w^{5} w^{7} w^{7} w^{3} w^{14} w^{10} w^{6}$ & \texttt{TCATAGCCGTGTAGTCGGAC} \\
			\hline
			226 & $w^{3} 0 w^{6} w^{11} w^{3} w^{3} w^{5} 1 0 w^{3}$ & \texttt{AGAAACCTAGAGCCTTAAAG} \\
			\hline
			227 & $w^{7} w^{5} 0 w^{8} w^{13} w^{13} w^{10} w^{9} w^{14} w^{2}$ & \texttt{GTCCAACGTGTGGGCATCGC} \\
			\hline
			228 & $w^{9} w^{2} w^{2} w^{13} w^{11} w^{11} w^{14} w^{1} w^{11} 0$ & \texttt{CAGCGCTGCTCTTCATCTAA} \\
			\hline
			229 & $w^{13} w^{9} w^{7} 0 w^{10} w^{10} w^{8} w^{12} w^{3} w^{12}$ & \texttt{TGCAGTAAGGGGCGGAAGGA} \\
			\hline
			230 & $w^{6} w^{3} w^{12} w^{3} w^{2} w^{2} w^{7} w^{10} w^{12} w^{7}$ & \texttt{ACAGGAAGGCGCGTGGGAGT} \\
			\hline
			231 & $w^{12} w^{11} w^{4} w^{4} 1 1 w^{9} w^{6} w^{5} w^{10}$ & \texttt{GACTTATATTTTCAACCCGG} \\
			\hline
			232 & $w^{11} w^{6} w^{10} w^{7} w^{1} w^{1} w^{2} w^{2} 1 w^{4}$ & \texttt{CTACGGGTATATGCGCTTTA} \\
			\hline
			233 & $1 1 w^{11} w^{9} 0 0 1 w^{8} w^{9} w^{9}$ & \texttt{TTTTCTCAAAAATTCGCACA} \\
			\hline
			234 & $0 w^{4} w^{9} w^{1} w^{4} w^{4} w^{12} w^{13} w^{13} w^{11}$ & \texttt{AATACAATTATAGATGTGCT} \\
			\hline
			235 & $w^{4} w^{8} w^{5} w^{14} w^{5} w^{5} w^{11} w^{5} w^{2} w^{1}$ & \texttt{TACGCCTCCCCCCTCCGCAT} \\
			\hline
			236 & $w^{1} w^{10} w^{1} 1 w^{8} w^{8} 0 w^{11} w^{4} w^{5}$ & \texttt{ATGGATTTCGCGAACTTACC} \\
			\hline
			237 & $w^{8} w^{14} w^{13} w^{6} w^{6} w^{6} w^{6} w^{4} w^{8} w^{14}$ & \texttt{CGTCTGACACACACTACGTC} \\
			\hline
			238 & $w^{2} w^{7} w^{14} w^{2} w^{12} w^{12} w^{1} w^{7} w^{1} w^{13}$ & \texttt{GCGTTCGCGAGAATGTATTG} \\
			\hline
			239 & $w^{10} w^{13} w^{8} w^{12} w^{9} w^{9} w^{4} 0 w^{7} 1$ & \texttt{GGTGCGGACACATAAAGTTT} \\
			\hline
			240 & $w^{5} w^{12} 1 w^{10} w^{14} w^{14} w^{13} w^{3} w^{6} w^{8}$ & \texttt{CCGATTGGTCTCTGAGACCG} \\
			\hline
			241 & $1 w^{2} w^{4} w^{6} w^{8} w^{8} w^{4} 1 w^{11} w^{7}$ & \texttt{TTGCTAACCGCGTATTCTGT} \\
			\hline
			242 & $0 w^{5} w^{10} w^{2} w^{5} w^{5} w^{13} w^{14} w^{14} w^{12}$ & \texttt{AACCGGGCCCCCTGTCTCGA} \\
			\hline
			243 & $w^{4} 0 w^{7} w^{12} w^{4} w^{4} w^{6} w^{1} 0 w^{4}$ & \texttt{TAAAGTGATATAACATAATA} \\
			\hline
			244 & $w^{1} w^{1} w^{12} w^{10} 0 0 w^{1} w^{9} w^{10} w^{10}$ & \texttt{ATATGAGGAAAAATCAGGGG} \\
			\hline
			245 & $w^{8} w^{6} 0 w^{9} w^{14} w^{14} w^{11} w^{10} 1 w^{3}$ & \texttt{CGACAACATCTCCTGGTTAG} \\
			\hline
			246 & $w^{2} w^{11} w^{2} w^{1} w^{9} w^{9} 0 w^{12} w^{5} w^{6}$ & \texttt{GCCTGCATCACAAAGACCAC} \\
			\hline
			247 & $w^{10} w^{3} w^{3} w^{14} w^{12} w^{12} 1 w^{2} w^{12} 0$ & \texttt{GGAGAGTCGAGATTGCGAAA} \\
			\hline
			248 & $w^{5} w^{9} w^{6} 1 w^{6} w^{6} w^{12} w^{6} w^{3} w^{2}$ & \texttt{CCCAACTTACACGAACAGGC} \\
			\hline
			249 & $w^{14} w^{10} w^{8} 0 w^{11} w^{11} w^{9} w^{13} w^{4} w^{13}$ & \texttt{TCGGCGAACTCTCATGTATG} \\
			\hline
			250 & $w^{3} w^{8} 1 w^{3} w^{13} w^{13} w^{2} w^{8} w^{2} w^{14}$ & \texttt{AGCGTTAGTGTGGCCGGCTC} \\
			\hline
			251 & $w^{7} w^{4} w^{13} w^{4} w^{3} w^{3} w^{8} w^{11} w^{13} w^{8}$ & \texttt{GTTATGTAAGAGCGCTTGCG} \\
			\hline
			252 & $w^{9} 1 w^{14} w^{7} w^{7} w^{7} w^{7} w^{5} w^{9} 1$ & \texttt{CATTTCGTGTGTGTCCCATT} \\
			\hline
			253 & $w^{13} w^{12} w^{5} w^{5} w^{1} w^{1} w^{10} w^{7} w^{6} w^{11}$ & \texttt{TGGACCCCATATGGGTACCT} \\
			\hline
			254 & $w^{6} w^{13} w^{1} w^{11} 1 1 w^{14} w^{4} w^{7} w^{9}$ & \texttt{ACTGATCTTTTTTCTAGTCA} \\
			\hline
			255 & $w^{12} w^{7} w^{11} w^{8} w^{2} w^{2} w^{3} w^{3} w^{1} w^{5}$ & \texttt{GAGTCTCGGCGCAGAGATCC} \\
			\hline
			256 & $w^{11} w^{14} w^{9} w^{13} w^{10} w^{10} w^{5} 0 w^{8} w^{1}$ & \texttt{CTTCCATGGGGGCCAACGAT} \\
			\hline
		\end{longtable}
	\end{exl}
	\begin{rmk}\label{better 2}
		There are precisely $16$ involutory $2\times 2$ matrices $D$ over $\mathbb{F}_{16}$ such that $D^{o4}=D$. Thus any full rank $2\times n$ matrix $M$ over $\mathbb{F}_{16}$ can be concatenated with $DM^{o4}J_n$ for any of the 16 $D's$ and the concatenated matrix will generate a linear code whose $\Phi$ image is a reversible DNA code. 
	\end{rmk}
	
	\section{Conclusion and Scope}
	In this work, we have developed a unified and significantly generalized framework for the construction and analysis of reversible linear and reversible DNA codes, extending and strengthening the approach followed in \cite{Oztas}. The central contribution lies in the introduction of a matrix-theoretic concatenation scheme based on involutory matrices, which substantially enlarges the class of admissible generator matrices that yield reversible structures. Unlike the approach in \cite{Oztas} that rely on more restrictive constructions, the present framework systematically uses products of involutory matrices to generate new families of codes while preserving reversibility. This not only broadens the algebraic scope of the theory but also leads to improved code parameters, particularly in terms of minimum distance. Due to this fact, it helps a lot in sustainable development.  This also provides structural flexibility in terms of generator matrix construction.
	
	The generalization has been carried out at two distinct but interconnected levels. First, at the level of linear codes over finite fields, we establish that concatenation using involutory matrices yields a rich class of reversible codes with enhanced parameters. The use of matrix products as a foundational tool allows for a more systematic proof strategy, in comparision to earlier codeword-wise  arguments given in \cite{Oztas}. This technique provides deeper insight into the algebraic structures governing reversibility, showing that the  involutory matrices and generator matrices can be systematically controlled to achieve desired properties. As a consequence, the construction naturally accommodates a wider variety of generator matrices, thereby overcoming limitations inherent in the previous approach.
	
	At the second level, these results are extended to DNA codes through the use of a DNA correspondence map in \cite{Oztas} and Hadamard power, which serve as a bridge between algebraic coding structures and nucleotide representations. By embedding the linear constructions into the DNA framework, we obtain reversible DNA codes that inherit the improved parameters of their algebraic counterparts. The Frobenius-based approach ensures compatibility with biological constraints while maintaining the algebraic properties of the construction.
	A key distinguishing feature of this work is the systematic use of matrix product techniques in both construction and proof. This approach is different from existing methods and offers a convenient algebraic tool for analyzing reversibility. It enables easier derivations, gives generalization, and provides a clear structural interpretation of the resulting codes. In particular, the involutory nature of the matrices plays a crucial role in ensuring reversibility, and the product structure allows for multiple constructions. This framework is expected to be adaptable to other classes of codes and may inspire further developments in algebraic coding theory.
	
	In addition to these constructive contributions, we address some problems in the theory. Specifically, we provide a solution to the open problem posed in \cite{Oztas}, regarding the characterization of reversible DNA codes as images of non reversible linear codes. Furthermore, we give a counterexample that falsifies one of the earlier characterizations of reversible-complement DNA codes. We also provide the correct characterization.
	
	The implications of these results extend beyond theoretical interest. Reversible and reversible-complement DNA codes play a crucial role in DNA-based data storage, where error resilience, hybridization control, and sequence stability are of foundational importance. The expanded class of generator matrices and improved code parameters obtained in this work can lead to more efficient encoding schemes, higher storage density, and enhanced robustness against synthesis and sequencing errors. Moreover, the algebraic flexibility of the proposed constructions may allow for the incorporation of additional biochemical constraints, such as GC-content balancing and avoidance of secondary structures, thereby increasing their practical applicability. In a broader context, the matrix-based framework may also find applications in other areas of coding theory and information science where symmetry and reversibility are desirable properties. As we can concatenate a given generator matrix with a large number of D-reflected matrices, we can use this process for cryptographic key development.
	
	In summary, this work advances the theory of reversible codes by introducing a flexible construction scheme based on involutory matrix concatenation. By generalizing existing results at both the linear and DNA levels, employing a novel matrix product approach, resolving an open problem, and correcting a previously held characterization, we provide a deeper and more accurate understanding of reversible coding structures. These contributions not only enrich the theory but also opens new possibilities for practical applications and future research.
	\section*{Acknowledgement} The second author is  grateful to the Council of Scientific and Industrial Research, India, for awarding the Junior Research Fellowship, file number: 09/0001(20883)/2025-EMR-I.
	
\end{document}